%% file: main.tex
\documentclass[11pt,dvipsnames]{article}

\usepackage{natbib}

\usepackage{amsthm}
\usepackage{amsmath}
\usepackage{refcount}
\usepackage{fullpage}

\usepackage{paralist}

\usepackage{appendix}
\usepackage{graphicx}

\usepackage{bm}

\usepackage{bbold}

\usepackage{algpseudocode}
\usepackage{verbatim}

\renewcommand{\phi}{\varphi}

\usepackage{todonotes}

\newcommand{\hide}[1]{ }

\renewcommand{\mathbf}{\bm}

\theoremstyle{plain}

\renewcommand{\include}{\input}

\usepackage{amsfonts}

\usepackage{tikz}
\usetikzlibrary{decorations.pathreplacing,angles,quotes}
\usetikzlibrary{arrows}

\usepackage{float}
\usetikzlibrary{patterns}

\usepackage{libertine}
\usepackage{geometry}
\usepackage{enumitem}
\PassOptionsToPackage{vlined, boxruled}{algorithm2e}
\usepackage{algorithm2e}

\usepackage{algpseudocode}

\usepackage{mleftright}
\usepackage{relsize}

\usepackage{caption}
\usepackage{subcaption}

\usepackage{bbm}

\usepackage{amsmath}
\usepackage{amsthm}
\usepackage{graphicx}
\usepackage{tabularx}
\newcolumntype{Y}{>{\centering\arraybackslash}X}

\usepackage{array}
\usepackage{soul}

\usepackage{mdframed}
\usepackage{xparse}
\usepackage{hyperref}
\hypersetup{colorlinks,linkcolor=blue,citecolor=blue,urlcolor=blue}

\usepackage[capitalise, noabbrev]{cleveref}

\usepackage{tcolorbox}
\usepackage{tikz,lipsum}
\usepackage{newtxmath}
\tcbuselibrary{skins,breakable}

\input{Sections/temp_commands_new}

\title{Beyond Monotone Delays for Multi-Level Aggregation}
\author{
Yossi Azar\thanks{Department of Computer Science, Tel Aviv University, Israel.
Email: azar@tau.ac.il}
\and
Liad Iluz\thanks{Department of Computer Science, Tel Aviv University, Israel.
Email: liadiluz@mail.tau.ac.il, liadiluz8@gmail.com}
}
\date{}

\begin{document}
\maketitle

\begin{abstract}
    In the online Multi-Level Aggregation Problem (MLAP), requests arrive over time and are associated with nodes of a given weighted rooted tree of depth~$D$. Each request must eventually be served by performing a service. Serving a request consists of selecting a rooted subtree that contains the request’s node, incurring a service cost equal to the total weight of the selected subtree. To reduce service costs, multiple requests may be served simultaneously by selecting a single rooted subtree that spans all of them. In addition, each request is associated with a penalty function that specifies the cost incurred when the request is served at a particular time. The objective is to minimize the total cost, consisting of both service costs and penalty costs.

    Most previous work on MLAP assumes monotone non-decreasing penalty functions, commonly referred to as delay functions. Only very recent results consider penalty functions that initially decrease and subsequently increase, and even then only for the special cases of depths $D=1$ and $D=2$, namely the Joint Replenishment Problem (JRP). 
    
    In this work, we extend previous results in two ways. First, we allow \emph{arbitrary penalty functions}, which may decrease and increase multiple times. Second, we study the general MLAP with \emph{arbitrary tree depth}~$D$ under these arbitrary penalty functions. 
    We present a randomized algorithm which is 
    $O\!\left(D \log {n} \log\!\left({n}DW\right)\right)$-competitive, where $W$ is the maximum service window among all penalty functions after normalizing the Lipschitz parameter of each penalty function to be 1 and normalizing the minimum positive edge weight incident to the root to be 1; and $n$ is the number of requests.
    We note that our algorithm runs in polynomial-time, and even for $D=1$ the problem admits hardness of approximation of $\Omega(\log n)$ for polynomial time algorithms.
    
    As mentioned above, prior to our work even for trees of depth $D=1,2$, non-monotone penalty functions have been studied only in special cases of functions that decrease and increase only once.
    In contrast, for such trees we obtain $O(\log n \log\!\left(nW\right))$-competitive algorithms for arbitrary non-monotone penalty functions. 
\end{abstract}

\setcounter{page}{1}
\input{Sections/1.introduction}

\input{Sections/2.preliminaries}
\input{Sections/3.discretization}
\input{Sections/4.reduction_imp}
\input{Sections/5.algorithm_imp}
\input{Sections/6.lower_bound}
\input{Sections/7.concluding_remarks}

\bibliographystyle{plain}
\bibliography{Sections/bibliography}

\input{Sections/8.apendix}
\end{document}

%% file: Sections/temp_commands_new.tex
\theoremstyle{plain}
\newtheorem{thm1}{Theorem}[section]
\theoremstyle{remark}
\newtheorem{pthm1}[thm1]{Theorem}
\theoremstyle{plain}
\newtheorem{lem1}[thm1]{Lemma}
\theoremstyle{plain}
\newtheorem{obs1}[thm1]{Observation}
\theoremstyle{plain}
\newtheorem{inv1}[thm1]{Invariant}
\theoremstyle{plain}
\newtheorem{cor1}[thm1]{Corollary}
\theoremstyle{definition}
\newtheorem{defn1}[thm1]{Definition}
\theoremstyle{plain}
\newtheorem{fact1}[thm1]{Fact}
\theoremstyle{remark}
\newtheorem{rem1}[thm1]{Remark}
\theoremstyle{plain}
\newtheorem{prop1}[thm1]{Proposition}
\theoremstyle{plain}
\newtheorem{asmp1}[thm1]{Assumption}

\ifx\proof\undefined
\newenvironment{proof}[1][\protect\proofname]{\par
\normalfont\topsep6\p@\@plus6\p@\relax
\trivlist
\itemindent\parindent
\item[\hskip\labelsep\scshape #1]\ignorespaces
}{%
\endtrivlist\@endpefalse
}
\providecommand{\proofname}{Proof}
\fi

\def\special{0}
\def\specialproof{0}
\def\specialdefinition{0}
\def\specialremark{0}
 \def\highlight{0}
 \def\sidebar{1}

\def\colorequations{0}

\input{Sections/env_defs}

%% file: Sections/env_defs.tex
\newenvironment{theorem}[1][]{%
\begin{thm1}[#1]%
}{\end{thm1}%
}
\newenvironment{lemma}[1][]{%
\begin{lem1}[#1]%
}{\end{lem1}%
}
\newenvironment{claim}[1][]{%
\begin{lem1}[#1]%
}{\end{lem1}%
}
\newenvironment{observation}[1][]{%
\begin{obs1}[#1]%
}{\end{obs1}%
}
\newenvironment{inv}[1][]{%
\begin{inv1}[#1]%
}{\end{inv1}%
}
\newenvironment{fact}[1][]{%
\begin{fact1}[#1]%
}{\end{fact1}%
}
\newenvironment{remark}[1][]{%
\begin{rem1}[#1]%
}{\end{rem1}%
}
\newenvironment{pthm}[1][]{%
\begin{pthm1}[#1]%
}{\end{pthm1}%
}
\newenvironment{corollary}[1][]{%
\begin{cor1}[#1]%
}{\end{cor1}%
}
\newenvironment{definition}[1][]{%
\begin{defn1}[#1]%
}{\end{defn1}%
}
\newenvironment{asmp}[1][]{%
\begin{asmp1}[#1]%
}{\end{asmp1}%
}
\newenvironment{proposition}[1][]{%
\begin{prop1}[#1]%
}{\end{prop1}
}%

\if 1
    \if 1\highlight
        \renewenvironment{theorem}[1][]{%
        \begin{mdframed}[nobreak=false,backgroundcolor=Aquamarine!60]\begin{thm1}[#1]%
        }{\end{thm1}\end{mdframed}%
        }
        \renewenvironment{lemma}[1][]{%
        \begin{mdframed}[nobreak=false,backgroundcolor=YellowGreen!60]\begin{lem1}[#1]%
        }{\end{lem1}\end{mdframed}%
        }
        
        \renewenvironment{observation}[1][]{%
        \begin{mdframed}[nobreak=false,backgroundcolor=Salmon!60]\begin{obs1}[#1]%
        }{\end{obs1}\end{mdframed}%
        }

    \fi
    \if 1\sidebar
        \tcolorboxenvironment{theorem}{blanker,breakable,left=4mm, before skip=10pt,after skip=10pt, borderline west={2mm}{0pt}{Aquamarine}}
        \tcolorboxenvironment{lemma}{blanker,breakable,left=4mm, before skip=10pt,after skip=10pt, borderline west={2mm}{0pt}{YellowGreen}}
        \tcolorboxenvironment{claim}{blanker,breakable,left=4mm, before skip=10pt,after skip=10pt, borderline west={2mm}{0pt}{YellowGreen}}
        \tcolorboxenvironment{observation}{blanker,breakable,left=4mm, before skip=10pt,after skip=10pt, borderline west={2mm}{0pt}{Salmon}}
        \tcolorboxenvironment{inv}{blanker,breakable,left=4mm, before skip=10pt,after skip=10pt, borderline west={2mm}{0pt}{Salmon}}
        \tcolorboxenvironment{fact}{blanker,breakable,left=4mm, before skip=10pt,after skip=10pt, borderline west={2mm}{0pt}{Salmon}}
        \tcolorboxenvironment{pthm}{blanker,breakable,left=4mm, before skip=10pt,after skip=10pt, borderline west={2mm}{0pt}{Salmon}}
        \tcolorboxenvironment{proposition}{blanker,breakable,left=4mm, before skip=10pt,after skip=10pt, borderline west={2mm}{0pt}{Goldenrod}}
        \tcolorboxenvironment{corollary}{blanker,breakable,left=4mm, before skip=10pt,after skip=10pt, borderline west={2mm}{0pt}{Mulberry}}
        \tcolorboxenvironment{asmp}{blanker,breakable,left=4mm, before skip=10pt,after skip=10pt, borderline west={2mm}{0pt}{Goldenrod}}
    \fi
\fi

\if 1\specialproof
    \if 1\highlight
        \expandafter\let\expandafter\oldproof\csname\string\proof\endcsname
        \let\oldendproof\endproof
        \renewenvironment{proof}[1][\proofname]{%
        \begin{mdframed}[nobreak=false,backgroundcolor=lightgray!60]\oldproof[#1]%
        }{\oldendproof\end{mdframed}}
    \else
        \if 1\sidebar
        \tcolorboxenvironment{proof}{
        blanker,breakable,left=4mm,
        before skip=10pt,after skip=10pt,
        borderline west={2mm}{0pt}{lightgray}}
        \fi
    \fi
\fi

\if 1\specialdefinition
    \if 1\highlight
        \renewenvironment{definition}[1][]{%
        \begin{mdframed}[innerbottommargin=0.1cm,innertopmargin=0.1cm,backgroundcolor=Apricot!60]\begin{defn1}[#1]%
        }{\end{defn1}\end{mdframed}%
        }
    \else
        \if 1\sidebar
        \tcolorboxenvironment{definition}{
        blanker,breakable,left=4mm,
        before skip=10pt,after skip=10pt,
        borderline west={2mm}{0pt}{Apricot}}
        \fi
    \fi
\fi

\if 1\specialremark
    \if 1\highlight
        
    \else
        \if 1\sidebar
        \tcolorboxenvironment{remark}{
        blanker,breakable,left=4mm,
        before skip=10pt,after skip=10pt,
        borderline west={2mm}{0pt}{Salmon}}
        \fi
    \fi
\fi

\if 1\colorequations
    \everymath{\color{MidnightBlue}}
\fi

%% file: Sections/1.introduction.tex
\section{Introduction}

The online \emph{Multi-Level Aggregation Problem} (MLAP) is a general aggregation problem that captures the trade-off between waiting in order to batch requests, thereby sharing an expensive resource, and serving requests at an optimal time to reduce time-dependent penalties. Every request arrives at a node of a given weighted rooted tree $\mathcal{T}$ and must eventually be served. A service corresponds to selecting a rooted subtree $\mathcal{T}' \subseteq \mathcal{T}$, where all requests located in $\mathcal{T}'$ can be served simultaneously. \emph{The service cost} or \emph{the cost of performing a service} is defined as the total weight of the edges in $\mathcal{T}'$. In addition, each request incurs a time-dependent penalty cost, which depends on the time at which the request is served. The objective is to minimize the total cost, defined as the sum of the service costs and the penalties incurred by all requests.

The depth of the tree $D$, is a key structural parameter of the problem. When $D=1$, MLAP is equivalent to the \emph{TCP Acknowledgment Problem} (TCP-AP), where the tree consists only of a root and a single child, and the edge weight represents the cost of performing an acknowledgment (also called single-item JRP). When $D=2$, MLAP is equivalent to the \emph{Joint Replenishment Problem} (JRP), which models multiple items ordered with a fixed order cost (the cost of the root) and with item-specific costs (the cost of the second level).

Most prior work on the various variants of the MLAP has focused on instances with monotone non-decreasing penalty functions, commonly referred to as \emph{delay functions} (see, e.g.,~\cite{Bi:15, AT:19, Bu:17} and more in Subsection~\ref{subsec:related_works}). More recent work has studied the Joint Replenishment Problem (JRP) under penalty functions that combine holding and backlog costs. In this setting, the penalty function first decreases (modeling holding cost) and then increases (modeling backlog cost), and is assumed to be identical across all requests~\cite{Mo:24}. Subsequent work extended these results to allow penalty functions that first decrease and then increase to differ across requests, while still achieving constant-competitive algorithms~\cite{AS:25, S:25}.

In contrast, we allow the penalty functions to exhibit arbitrary behavior, possibly decrease and increase multiple times. This has not been considered even for the TCP-AP (single-item JRP). Moreover, we consider the depth of the tree~$D$ as an arbitrary parameter and not restricted only to a constant $D$, rather than restricting attention to the constant-depth case~($D=2$) as in the JRP. For such an arbitrary $D$, only monotone non-decreasing penalty functions (delay functions) have been considered. Since the common term \emph{delay} implies monotonicity, we instead adopt the more general term \emph{penalty function}, and use \emph{penalty} to denote the cost incurred when serving a request at a given time.

To this end, we introduce a new problem, the \emph{online Incremental Multicast Problem} (IMP). This problem is a natural generalization of the classical online Multicast Problem (MP). Unlike the multicast problem, where the input trees are fixed in advance, in the incremental multicast problem the trees are initially empty and may grow over time as requests arrive.

In the standard online multicast problem on trees, the input consists of a collection of weighted rooted trees. Requests (clients) arrive online, and each request is associated with at most one node in each tree. Upon its arrival, the algorithm must immediately cover the request by adding edges to the solution such that at least one path from the request’s associated node to the root of its tree is included. The algorithm may add new edges to the solution over time, but cannot remove edges that were selected earlier. The objective is to minimize the total cost of all edges included in the solution, while ensuring that all requests are covered.

We then establish two key reductions that transform an instance of MLAP with arbitrary penalty functions into a corresponding instance of IMP. Finally, to obtain a complete algorithmic solution for our problem, we show that the randomized algorithm for the online multicast problem introduced by Alon et al.~\cite{Al:04} can be naturally adapted to solve the incremental multicast problem (IMP). We rewrite both the algorithm and its analysis to fit the context of this new variant. Combined with our reductions, this yields a full randomized algorithm for MLAP with arbitrary penalty functions.

Despite allowing penalty functions that may decrease and increase multiple times, our framework remains highly relevant to real-world applications, including classical delay-based models. For example, in supply chain and inventory management, where the JRP is commonly applied, the urgency of serving a request may vary over time. A retailer preparing for a seasonal sales period may face a penalty that increases as peak demand approaches, but then decreases afterward due to discounts, excess inventory or supplier flexibility.

\begin{figure}[t]
    \centering
    \includegraphics[width=0.75\linewidth]{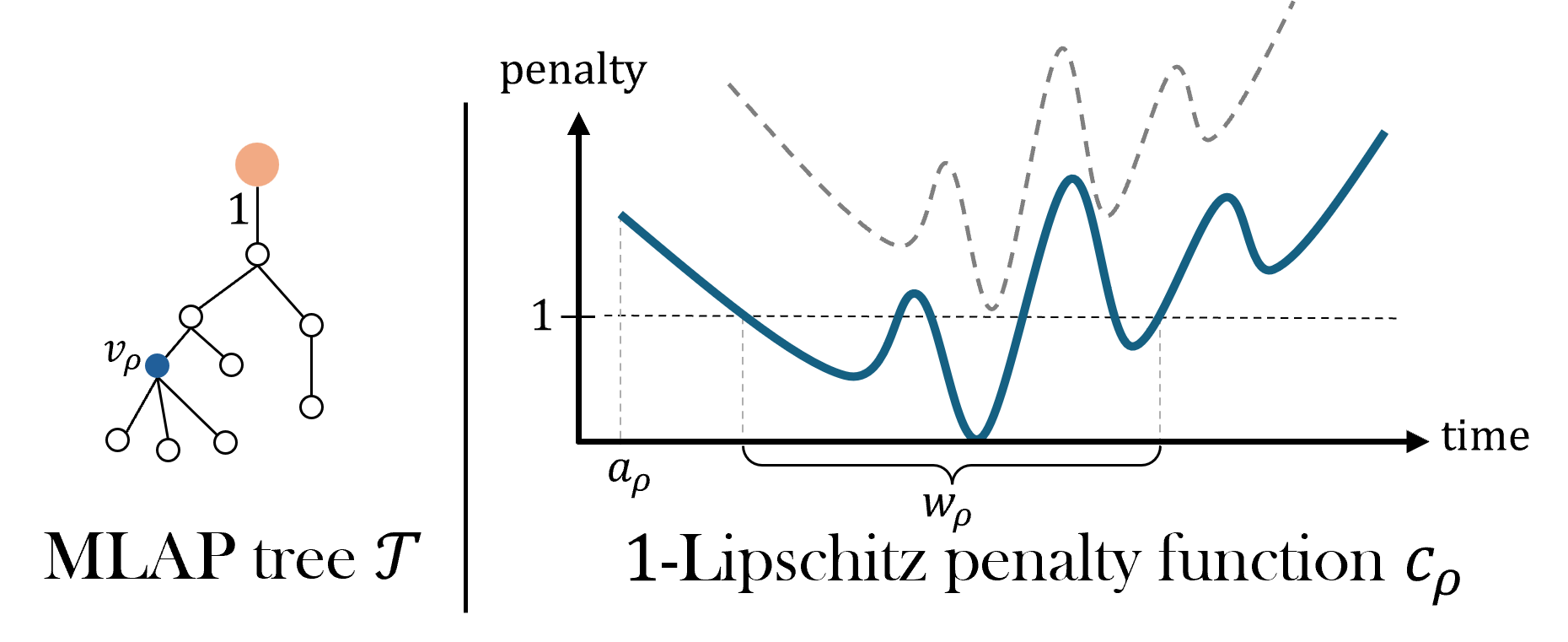}
    \caption{Request $\rho$ arrives at time $a_\rho$, is associated with a node $v_\rho$, and has a $1$-Lipschitz penalty function $c_\rho$, obtained by: 
        (i) dividing by the minimum positive edge incident to the root (making it to be $1$), 
        (ii) shifting $c_\rho$ so its minimum is $0$, and 
        (iii) rescaling to make it $1$-Lipschitz. 
        The dashed gray line represents the penalty function before normalization.
The service window $w_\rho$ is the interval between the earliest time with $c_\rho(t) \le 1$ and the latest time with $c_\rho(t) \ge 1$.}
    \label{fig:problem_def}
\end{figure}

\subsection{Our Results}
Let $n$ be the number of requests that arrive during the algorithm and let $D$ be the depth of the problem tree $\mathcal{T}$. Assume WLOG that all edge weights are positive (zero weight edges can be contracted), and that the root has only one connected node (otherwise, split the tree into a collection of disjoint trees, and solve the problem on each of them, as every two requests that are related to different trees can not share a service). Then, we normalize the weight of this only edge that incident to the root to be $1$, by normalizing the values of the penalty function using this value. 

We further assume, without loss of generality, that each penalty function attains the value $0$. Otherwise, we shift the penalty function by subtracting its minimum value. This transformation does not affect the competitive-ratio analysis, since the same constant is subtracted from both the algorithm's cost and the optimal cost. Consequently, the new competitive ratio is at least as large as the original one. We then normalize the penalty function to be $1$-Lipschitz\footnote{A function $f$ is called $L$-Lipschitz if for all $x,y$: $|f(x)-f(y)| \le L|x-y|$.} by appropriately rescaling the time axis.
Let $w_\rho$ denote the \emph{service window} of request $\rho$, defined as the interval from the earliest time at which the penalty (after normalization) is at most $1$, to the latest time at which it is at least $1$. Let $W$ be the maximum service window over all requests.

The parameter $W$ need not be known in advance, and it captures the complexity of the problem: larger values of $W$ correspond to penalty functions that originally vary more quickly or allow longer service horizons. Figure \ref{fig:problem_def} illustrates these definitions. 

Finally, we briefly define a \emph{2-decreasing rooted tree} as a weighted rooted tree in which the weight of every edge is at most half the weight of its parent edge.

The following results are given as a worst case upper bound analysis of the competitive ratio between the value of our algorithm to the value of the optimal solution, which know the whole input sequence in advance.

\begin{itemize}
    \item For the online MLAP with arbitrary penalty functions on $2\text{-decreasing}$ trees, we design a  randomized algorithm which is ${O}\left(\log n\log\!\left(nDW\right)
    \right)$-competitive.
    \item There exists a randomized algorithm for the online MLAP with arbitrary penalty functions for general trees that is
    ${O}\left(D \log n\log\!\left( nDW\right)\right)$-competitive.
    \item For constant $D$ (in particular, for the JRP with arbitrary penalty functions), we obtain a randomized algorithm which is  ${O}\left(\log n\log\!\left( nW\right)\right)$-competitive.
\end{itemize}

{\noindent\bf Remark. } Although the parameter $n$ may be large, we also show that it can be replaced by $\bar{n}$ in all the above results, which is typically much smaller than $n$. Here, $\bar{n}$ denotes the maximum number of requests that arrive within the same window of length $W$.

We note that \cite{S:25} showed that an $\alpha$-approximation algorithm for the offline single-item JRP ($D=1$) with arbitrary penalty functions implies the same approximation for the offline set-cover problem. This implies a hardness of $\Omega(\log n)$ for any polynomial time algorithm, even for a single-item JRP, unless $P=NP$. However, if we allow exponential time (or unbounded time) algorithms, the best lower bound for online deterministic algorithms is $2$ for $D=1$ and $2.754$ for $D
\ge2$. We improve this result for non-monotone penalty functions. Specifically, we obtained a lower bound of $4$ for $D=1$ with penalty functions with values in $[0,1]$. Note that this lower bound is higher than the upper bound of the JRP with penalty functions that decrease and subsequently increase considered in \cite{S:25}, which separates the two problems.

\subsection{Our Techniques}

Our solution is built upon three reductions, which can be illustrated by the following flow diagram:
\[
\text{MLAP}(n,D) \Longrightarrow \text{MLAP}_2(n,D) 
\Longrightarrow \text{MLAP}_\text{2-Discrete}(n,D,T) 
\Longrightarrow \text{IMP}(n,m)
\]

The parameter $n$ is the number of requests in all problems, and in the MLAP variants, $D$ refers to the depth of the tree. The first reduction transforms a general $\text{MLAP}(n,D)$ instance with an arbitrary tree into $\text{MLAP}_2(n,D)$ -- the same problem with $2$-decreasing trees. This step is standard in many MLAP solutions and incurs a multiplicative loss of a factor $2D$ in the competitive ratio. The second reduction takes an instance of $\text{MLAP}_2$ with arbitrary penalty functions defined over a continuous timeline and discretizes them into a set of $T$ candidate service points in time, for appropriate parameter $T$. The resulting problem is called $\text{MLAP}_\text{2-Discrete}$. The third and final reduction converts $\text{MLAP}_\text{2-Discrete}$ into an instance of the Incremental Multicast Problem (IMP), a generalization of the multicast problem that we introduce in this paper. Here, the parameter $m$ denotes the number of edges in the final IMP graph.
We show the following results towards a full randomized algorithm for the MLAP with arbitrary penalty function:
\begin{enumerate}
    \item \textbf{Reduction I: } $\text{MLAP}(n,D) \;\le^{2D}\; \text{MLAP}_{2}(n,D)$.
    \item \textbf{Reduction II: } $\text{MLAP}_{2}(n,D) \;\le\; \text{MLAP}_\text{2-Discrete}(n,D,T)$ with $ |T| = O\!\left(nW\right)$.
    \item \textbf{Reduction III: } $\text{MLAP}_\text{2-Discrete}(n,D,T) \;\le\; \text{IMP}(n,m)$ with $ m = O\!\left(nD|T|\right)$.
    \item \textbf{Algorithm I: } Designing a ${O}\left(\log n \log m\right)$-competitive randomized algorithm for the IMP, based on the randomized algorithm of Alon et. al.~\cite{Al:04} for the online multicast problem.   
\end{enumerate}

Putting all the results together, we get a randomized algorithm for the online MLAP with arbitrary penalty functions and arbitrary trees which is ${O}\left( D\log n\log(nDW) \right)$-competitive. 

\subsection{Related Works}
\label{subsec:related_works}

\paragraph*{The Joint Replenishment Problem (JRP)} models a retailer managing inventory for multiple items, where requests (demands) arrive in an online fashion over a continuous timeline. Orders incur 3 types of costs: a fixed joint cost, item-specific costs, and a time-dependent cost for serving a particular request at a specific time. The objective is to minimize the total cost. The JRP is a special case of the Multi-Level Aggregation Problem (MLAP) where $D=2$. A simpler variant is the single-item JRP, where there is only one item to be ordered. This problem is equivalent to the \emph{Dynamic TCP Acknowledgment Problem (TCP-AP)} ($D=1$).  

Classical work on JRP distinguishes two variants based on time-dependent costs: \emph{make-to-stock}, where items can be held with per-unit holding cost, and \emph{make-to-order}, where holding is disallowed and backlog costs are incurred. While equivalent in the offline setting, they differ online.
For the online make-to-order JRP, Brito et al.~\cite{BR:12} presented a 5-competitive algorithm. Buchbinder et al.~\cite{Bu:13} improved this result giving a 3-competitive algorithm using the primal-dual approach and proved a lower bound of $2.64$, which was later improved to $2.754$ for any deterministic algorithm by Bienkowski et al.~\cite{Bi:14}. For the single-item JRP (TCP-AP), Dooly et al.~\cite{Doo:98, Doo:01} gave a tight 2-competitive deterministic algorithm, Seiden~\cite{Se:00} proved an $e/(e-1)$ lower bound for randomized algorithms, and Karlin et al.~\cite{Ka:01} matched it with a randomized algorithm that is $e/(e-1)$-competitive. For the online make-to-stock JRP, Bienkowski et al.~\cite{Bi:14} designed a 2-competitive algorithm in the special case where holding costs are zero, and showed that this bound is optimal.  

More recently, research has shifted toward more general models. Moseley et al.~\cite{Mo:24} studied the case involving both holding and backlog costs. In their formulation, ordering too early incurs holding costs, while ordering too late incurs backlog costs. Thus, each request has a time-dependent cost function that first decreases (reflecting holding costs) and then increases (reflecting backlog costs), resulting in a unique time at which the cost of serving the request is minimized. They gave a deterministic $30$-competitive algorithm for the special case where all requests share the same time-dependent function. Their results were later extended in \cite{AS:25, S:25} to handle any such functions that may be different among requests, and also obtained constant-competitive algorithms.

Our work generalizes these results in two ways: (1) we consider arbitrary penalty functions that may exhibit multiple increases and decreases over time; and (2) we analyze the full MLAP under this model, rather than restricting to the JRP.

\paragraph*{The Multi-Level Aggregation Problem (MLAP)} was introduced by Bienkowski et al.~\cite{Bi:15} as a general model that captures the trade-off between delaying requests to benefit from aggregation, and serving them immediately to avoid high waiting costs. All potential requests are located at the nodes of a given rooted tree of depth $D$, and serving a set of requests corresponds to selecting a rooted subtree that spans those requests. Bienkowski et al.~\cite{Bi:15} presented an $O(D^4 2^D)$-competitive deterministic algorithm, which was later improved to $O(D^2)$-competitive deterministic algorithm by Azar and Touitou~\cite{AT:19}. Azar and Touitou in \cite{AT:20} also show $O(\log k)$-competitive deterministic algorithm for this problem where $k$ is the number of nodes in the tree. For the special case of MLAP with deadlines (zero cost follows by infinity cost after the deadline), Buchbinder et al.~\cite{Bu:17} gave an $O(D)$-competitive deterministic algorithm. For the special case of a path, Bienkowski et al. \cite{Bi:13} gave an algorithm which is 5-competitive and showed a lower bound of 3.618, which was improved later to 4 by \cite{More:Bi:21}. For both variants, the deadline and the delay, the best known lower bounds are only constant (as achieved under the JRP model). Additional variants of MLAP have been studied in \cite{More:Bi:21, More:Bo:20, More:Che:16, More:Mar:24, More:Mc:21, More:Nag:16}. Nevertheless, all of these results assume monotone non-decreasing penalty functions (i.e., delay functions), whereas our work considers arbitrary functions that may increase and decrease multiple times over time.
\vspace{0.2cm}

\paragraph*{The Multicast Problem (MP) in Trees.} Given a collection of weighted rooted trees, and a set of clients (requests) that are associated with at most one node in each tree. The objective is to select a minimum-weight collection of rooted subtrees such that all clients are served; that is, for every client, the solution contains at least one rooted subtree that spans at least one of the nodes that associated with this client.
In the online version, clients arrive one by one, and upon the arrival of a client, the algorithm must immediately serve it. The service cost is defined as the sum of the weights of all edges in the chosen subtree. The algorithm may add new edges to its current solution but cannot remove any previously chosen edges. The goal is to minimize the total serving cost over time.

Alon et al.~\cite{Al:04} described a randomized algorithm for the online MP in trees that achieves a competitive ratio of $O\!\left(\log n \cdot \log m\right)$, where $n$ is the number of clients that arrive during the execution of the algorithm and $m$ is the total number of edges across all trees. In a subsequent paper, Alon et al.~\cite{Al:09} established a nearly matching lower bound of
$
\Omega\!\left(\frac{\log n \log m}{\log\log n + \log\log m}\right),
$
for any deterministic algorithm, via a reduction from the online set cover problem, which is a special case of the MP.

\paragraph{The  Non-Metric Facility Location (NMFL) problem} is a special case of the multicast problem (MP) where all trees have depth $2$. The root (facility) connects to an intermediate node by an edge whose weight equals to the facility setup cost, and the intermediate node connects to leaves (clients), with edge weights corresponding to the clients’ connection cost to that facility. Alon et al.~\cite{Al:04} presented a randomized $O(\log |F| \log |A|)$-competitive algorithm for the NMFL, where $F$ is the set of facilities and $A$ is the set of all arriving clients. By applying the deterministic online set cover algorithm of Alon et al.~\cite{Al:09}, one obtains a deterministic algorithm which is $O\!\left((\log |C| + \log |F|)(\log |C| + \log\log |F|)\right)$-competitive, where $C$ is the set of all potential clients. Bienkowski et al.~\cite{Bi:21} improved this, giving an $O\!\left(\log|F| \cdot (\log|C| + \log|F|)\right)$-competitive deterministic algorithm, which is asymptotically better when $|F| \ll |C|$.

The conceptual approach used in our paper to generalize from the JRP to the MLAP is inspired by the relationship between the NMFL and the MP, where NMFL corresponds to the special case of trees of depth~2.

\subsection{Paper Organization}
Section \ref{sec:preliminaries} introduces notation and definitions for MLAP with arbitrary penalty functions and the Incremental Multicast Problem (IMP) in trees, along with a standard reduction used in MLAP variants.
Section \ref{sec:discretization} presents a reduction from continuous-time MLAP with arbitrary penalties to a discrete version; this reduction may be of independent interest.
Section \ref{sec:reduction_imp} gives the final reduction from discrete MLAP to IMP. Section \ref{sec:algorithm} adapts the algorithm for MP of Alon et al.~\cite{Al:04} to our IMP variant.
%
%
Section \ref{sec:lower_bound} establishes a lower bound for this problem even for the case of $D=1$.
Finally, Section \ref{sec:conculding-remarks} concludes with directions for future research.


%% file: Sections/2.preliminaries.tex
\section{Preliminaries}
\label{sec:preliminaries}

\subsection{The MLAP with arbitrary penalty functions}
The input of the MLAP consists of a weighted tree $\mathcal{T} = (V,E)$ with positive edge weights (or costs) $c : E \rightarrow \mathbb{R}^+$, rooted at a node $r \in V$ and having depth~$D$. Each request is associated with a non-root node of the tree, representing the request’s type. Serving a request means selecting a subtree rooted at~$r$ that contains the node of the request. Multiple requests may be aggregated and served simultaneously by selecting a single rooted subtree that spans all of them. The \emph{service cost} or \emph{the cost of performing a service}, is the sum of all weights in the selected subtree and is independent of the time at which the service is performed.

In addition, each request is associated with a penalty function that is revealed upon its arrival (clairvoyant model) and specifies the cost of serving the request at every point in continuous time. This function is arbitrary and may increase or decrease multiple times.
The algorithm may perform additional services over time, but it can not cancel or modify a service once executed. The objective is to minimize the total cost: service costs plus penalty costs, while ensuring that every request is eventually served.


\paragraph{Notation and Formal Definition.}
We normalize the cost of the unique edge incident to the root to be 1 (by scaling all weights). For every node $v \in V$ that is not the root $r$, we denote by $c(v)$ the cost of the edge between $v$ and its parent. For any subtree $\mathcal{T}' \subseteq \mathcal{T}$, we write $c(\mathcal{T}')$ for the total weight of all edges in $\mathcal{T}'$.

Requests associated with different nodes arrive over time. Each request $\rho$ is specified by a tuple $(v_\rho, a_\rho, c_\rho)$, where $v_\rho\in V$ is the node at which the request occurs, $a_\rho\in \mathbb{R}^+$ is the arrival time, and $c_\rho : [a_\rho, \infty) \to \mathbb{R}^+$ is the penalty function of the request. It means that serving request $\rho$ at time $t$ incurs a penalty of $c_\rho(t)$.

We normalize the penalty function to be $1$-Lipschitz by stretching the time. Since the penalty function may increase and decrease over time, we define the \emph{service window} of a request $\rho$, denoted by $w_\rho$, in order to bound the time period during which the request may be served. Specifically, $w_\rho$ is the interval from the first time at which the function attains a value at most $1$ to the last time at which it attains the value $1$. We assume that this window is finite and denote by $W$ the maximum length of a service window over all requests. Note that the values of these parameters need not be known in advance. Figure \ref{fig:problem_def} illustrates the arrival of request $\rho$.

\paragraph{Feasible Solution.}
Let $R$ be a sequence of requests, each associated with exactly one node of $\mathcal{T}$.  
A solution to the MLAP consists of a sequence of services (rooted subtrees)  
$\mathcal{T}_1, \mathcal{T}_2, \ldots, \mathcal{T}_\ell \subseteq \mathcal{T}$,  
where each service $\mathcal{T}_i$ is performed at time $t_i$.  Each service $\mathcal{T}_i$ has a corresponding set of requests $\mathcal{R}_i \subseteq R$ that are served by it. The solution is \emph{feasible} if both conditions hold: (1) for every $i = 1,\ldots,\ell$ and every $\rho \in \mathcal{R}_i$: $t_i \ge a_\rho$ and $v_\rho \in \mathcal{T}_i$; 
and (2) every request $\rho \in R$ must eventually be served, i.e., there exists some $i$ such that $\rho \in \mathcal{R}_i$.

\paragraph{The Goal.} Minimizing the total cost:
$
\sum_{i=1}^{\ell}{\left(c\left(\mathcal{T}_i\right) + \sum_{\rho \in \mathcal{R}_i}{c_\rho\left(t_i\right)}\right)}
$.

\paragraph{The MLAP with discrete penalty functions.}
In this problem the penalty functions are discrete, i.e., each request $\rho$ is associated with a discrete penalty function $c_\rho : T_\rho \rightarrow \mathbb{R}^+$, where $T_\rho$ is a finite set of all potential time units to serve  request $\rho$. We denote by $T$ the set of all potential time units to serve all the requests, i.e., $T=\bigcup_{\rho\in R}{T_\rho}$.

\subsection{Reducing MLAP on an arbitrary tree to MLAP on 2-decreasing trees}
In this section we show the first reduction. Since this is a very common reduction, we briefly describe it and refer to other works that extend it.

\begin{definition}
    Fix $\alpha \ge 1$.  
    A weighted tree $\mathcal{T}$ is called an \emph{$\alpha$-decreasing tree} if for every three consecutive nodes $w, v, u$, where $w$ is the parent of $v$ and $v$ is the parent of $u$, it holds that $c(v) \ge \alpha \cdot c(u)$.
\end{definition}

Following a standard reduction described in \cite{Bi:15, Bu:17, AT:19}, we construct a new graph as follows:  
for every node $u \in V$, except the root $r$, the parent of $u$ is defined as the closest ancestor $v$ along its only path to the root, such that $c(v) \ge 2c(u)$; if no such node exists, the parent is set to be $r$.  
This construction is well defined and yields $2$-decreasing trees. It is straightforward to verify that any algorithm for MLAP on $2$-decreasing trees can be transformed into an algorithm for MLAP on arbitrary tree, losing a factor of $2D$ in the approximation ratio (offline) or in the competitive ratio (online). For full details, see for instance Section~3 in Bienkowski et al. \cite{Bi:15}.


\subsection{The Incremental Multicast Problem (IMP)}

\begin{figure}[t]
    \centering
    \includegraphics[width=0.8\linewidth]{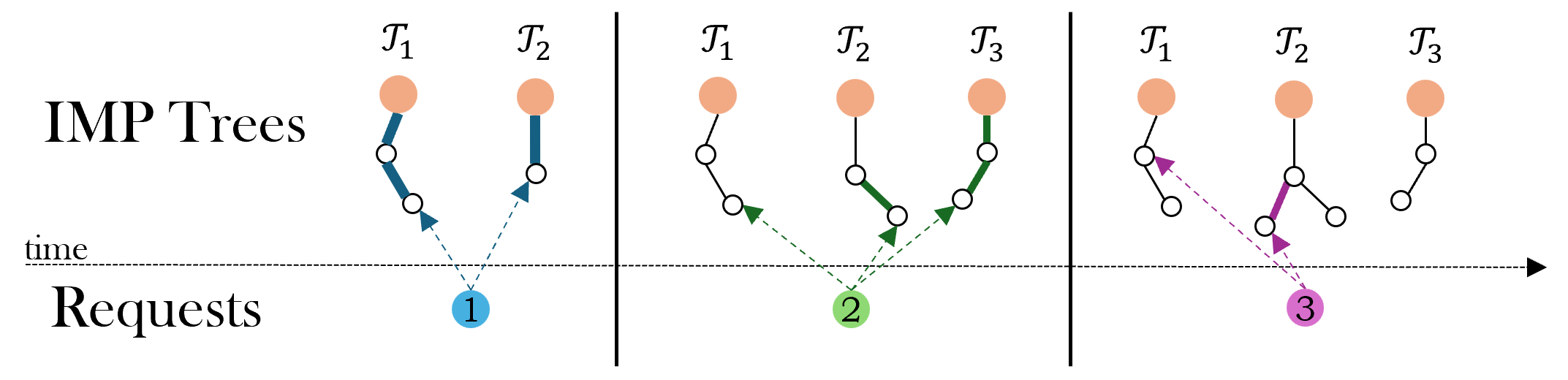}
    \caption{The figure shows three sequential requests. Each request (at the bottom) points to its associated vertices ($S_\rho$). Bold edges indicate newly added edges ($\{\mathcal{P}_v\}_{v\in S_\rho}$) and their costs ($c_\rho$).
The first request connects to two new vertices in different trees via paths of lengths 2 and 1. The second request involves three vertices: one already in the tree (length 0), one connected by a single new edge, and one connected by a path of length 2 that creates a new tree. The third request involves two vertices: one existing (length 0) and one connected by a single new edge (length 1).}
    \label{fig:IMP}
\end{figure}

We now introduce a new variant and a generalization of the online multicast problem. In contrast to the classical online multicast problem, in the incremental multicast problem the trees are not known in advance and may be incrementally constructed as requests arrive.

Let $R$ be the sequence of requests and let $\mathcal{T} = (V,E)$ be a family of rooted trees. Initially, both $R$ and $\mathcal{T}$ are empty. A new tree is introduced, together with its root, when the first request that is associated with this tree arrives, and the tree then grows incrementally as additional requests connect to it. Formally, each request~$\rho$ is represented by the tuple $(S_\rho, \{\mathcal{P}_v\}_{v\in S_\rho}, c_\rho)$, where:
\begin{itemize}
    \item $S_\rho$ is a set of nodes (already in $V$ or new nodes) that are associated with the request. Each node $v \in S_\rho$ belongs to at most one tree in $\mathcal{T}$.

   \item $\{\mathcal{P}_v\}_{v \in S_\rho}$ is a collection of paths in $\mathcal{T}$.
    For each node $v \in S_\rho$, the path $\mathcal{P}_v$ lies in some unique tree $\mathcal{T}' \in \mathcal{T}$. This path connects the node $v$ to an existing node in $\mathcal{T}'$ or initiates a new rooted tree (a path from the node to the new root). Each such a path consists of only new edges, i.e., edges that were not previously present in this tree, and adding $\mathcal{P}_v$ to $\mathcal{T}'$ preserves the property that $\mathcal{T}'$ is a tree.

    \item $c_\rho : \bigcup_{v \in S_\rho} E(\mathcal{P}_v) \to \mathbb{R}^+$ is a cost function on edges assigning a positive cost to every new edge introduced by the request.
\end{itemize}
Figure~\ref{fig:IMP} illustrates the incremental construction process in the IMP.

\paragraph{Feasible Solution.} Let $T_\rho$ denote the set of all roots that are associated with the nodes in $S_\rho$. When a request $\rho$ arrives, it induces a unit demand $D_\rho = \left(S_\rho, T_\rho\right)$. This means that at any time, the solution must contain at least one path from some node $v \in S_\rho$ to the root of its corresponding tree from $T_\rho$. The algorithm is allowed to add new edges to the solution, but never remove edges once they have been added.

\paragraph{The Goal.} The cost of the algorithm is the total cost of all edges in the final solution, and the objective is to minimize it.

\paragraph{The Fractional Model.}  
In the fractional model, the algorithm may assign a fractional weight $w_e \in [0,1]$ to any edge $e \in E$. These fractions may increase over time, but they can never decrease. A feasible solution is an assignment of weights $\{w_e\}_{e \in E}$ such that, at any moment and for every request $\rho$, the total flow from $S_\rho$ to $T_\rho$ is at least $1$. Notice that since the fractional weights never decrease, the feasibility conditions of all previously arrived requests remain satisfied. The cost of the fractional solution is the weighted sum of edge costs: $\sum_{e \in E} w_e c_e$.

\paragraph*{Simple Lower Bounds.}
Since the NMFL is a special case of the MP, the NMFL without knowing the facility-client graph in advance is a special case of the IMP. As Bienkowski et al. mentioned in \cite{Bi:21}, any \textbf{deterministic algorithm} for the online NMFL without knowing the facility-client graph in advance is at least $m$-competitive, where $m$ is the final number of edges. Therefore, this lower bound holds also for the IMP.
A similar argument, which has appeared in several previous works (e.g.,~\cite{Bu:09, Al:04}), shows that any \textbf{randomized algorithm} for this problem admits a lower bound of $\Omega(\log m)$. To illustrate this, consider the fractional model, which lower bounds any randomized algorithm.
The first request is connected to $m$ trees, each consisting of a single edge. The algorithm must assign weights to the edges so that the total weight is at least~$1$. Partition the trees into two arbitrary groups of equal size. At least one of these groups must have total weight at least~$\tfrac{1}{2}$. We discard this group and do not use it anymore.
The next request is then connected only to the remaining half of the trees. We repeat this procedure, halving the number of available trees at each step, until no trees remain. At each step, the algorithm incurs a cost of at least~$\tfrac{1}{2}$ from the discarded group. Hence, the total cost paid by the algorithm is $\Omega(\log m)$. In contrast, the optimum solution pays only~$1$, by serving the request using the last remaining tree.

%% file: Sections/3.discretization.tex
\section{Reduction to MLAP with discrete arbitrary penalty functions}
\label{sec:discretization}
\begin{figure}[t]
    \centering
    \includegraphics[width=0.6\linewidth]{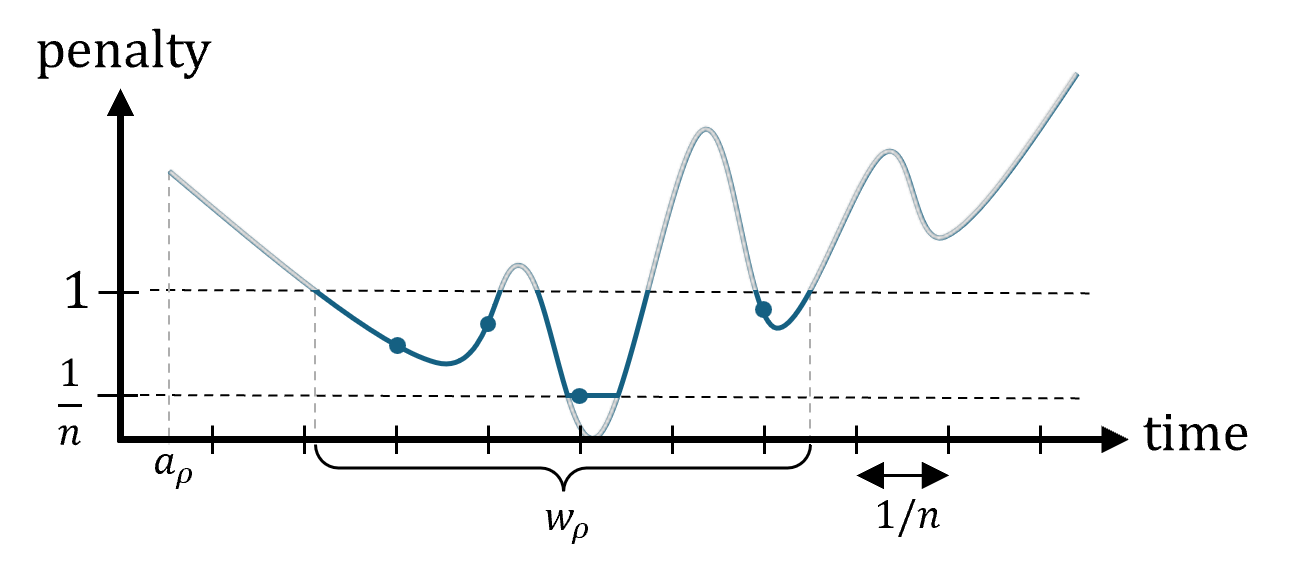}
    \caption{First, all values below $\frac{1}{n}$ get value $\frac{1}{n}$. Then, the discrete penalty function is defined at the points corresponding to the global intervals of length $\frac{1}{n}$ that lie within the service window and have values of at most $1$ (shown as blue circles).}
    \label{fig:discrete_function}
\end{figure}

In this section, we describe the second reduction from the MLAP on 2-decreasing trees with maximum depth $D$ and $n$ requests ($\text{MLAP}_2(n,D)$) to $\text{MLAP}_\text{2-Discrete}(n,D,T)$, the same problem with discrete arbitrary penalty functions with the set $T$ of all potential service time units. This technique is general and can be applied to other offline and online problems involving arbitrary penalty functions.

By Lipschitzness condition and bounded service window, we can assume that each penalty function has a global minimum point (with zero value as we mention earlier). For simplicity, we make two additional assumptions that will be removed at the end of this section: (1) all $n$ requests arrive in the time interval $[0,W]$ (see Subsection \ref{subsec:window}); and (2) the parameters $n$ and $W$ are known in advance (see Subsection \ref{subsec:doubling}). Finally, we note that throughout this section we work in the setting of MLAP on 2-decreasing trees with maximum depth $D$. The main result of this section:

\begin{theorem}[Discretization Step] Assume there is an algorithm $\mathcal{A}$ for the $\text{MLAP}_\text{2-Discrete}(n,D,T)$ which is an $\alpha$-approximation (offline) or competitive (online). Then, there is an algorithm $\mathcal{B}$ for the original $\text{MLAP}_2(n,D)$ which is an $O(\alpha)$-approximation (offline) or competitive (online), with $|T|=O\left(nW\right)$. Moreover, both algorithms do not need to know the parameters $n$ and $W$ in advance.
\end{theorem}

We start with the offline version and then we will show that it can be maintained in online fashion, using a two-step reduction:
\begin{enumerate}
    \item Reduce to \textbf{non-zero} penalty functions $\left[a_\rho, \infty\right) \rightarrow \left[\frac{1}{n},\infty\right)$, incurring a factor of 2 loss in approximation.
    \item Reduce to \textbf{discrete} penalty functions $T_\rho \rightarrow \left[\frac{1}{n}, 1 \right]$, where $T_\rho$ is the potentially discrete set of service time unites for $\rho$, incurring a factor of $6$ loss in approximation.
\end{enumerate}

Overall, we get a reduction from arbitrary penalty functions $\left[a_\rho, \infty\right) \rightarrow \left[0, \infty\right)$ to discrete penalty functions $T_\rho \rightarrow \left[\frac{1}{n},1\right]$, incurring a factor of $12$ loss in approximation. To show any of those reductions, we denote by $\text{ALG}_{\text{org}}$, $\text{ALG}_{\text{new}}$, $\text{OPT}_{\text{org}}$ and $\text{OPT}_{\text{new}}$ the values of the algorithm and the optimum on the original problem and on the new problem respectively. In each step, we denote by $c_\rho$ the original penalty function of a request $\rho$, and by $\hat{c}_\rho$ its new penalty function after applying the step. Figure \ref{fig:discrete_function} illustrates the obtained discrete function.




\paragraph*{Step 1: Non-zero values}
Define a new \textbf{non-zero penalty function} by setting every value less than $\frac{1}{n}$ to be $\frac{1}{n}$ and all other values remain the same: $\hat{c}_\rho = \max\{c_\rho,\; \frac{1}{n}\}$.

\begin{lemma}
Any $\alpha$-approximation (or competitive) algorithm for the problem using the new non-zero penalty functions is a $2\alpha$-approximation (or competitive) algorithm for the original penalty functions.
\end{lemma}
\label{lemma_first_step}

\begin{proof}
Since $\hat{c}_\rho(t) \ge c_\rho(t)$ for every $t$, it holds that $\text{ALG}_{\text{org}} \le \text{ALG}_{\text{new}}$. On the other hand, by keeping all service points as in $\text{OPT}_{\text{org}}$, the new optimum pays at most an additional $\frac{1}{n}\cdot n=1$. Finally, since $\text{OPT}_{\text{org}} \ge 1$ for serving at least one request, we get
$\text{OPT}_{\text{new}} \le \text{OPT}_{\text{org}} + 1 \le 2 \cdot \text{OPT}_{\text{org}}$.
\end{proof}

\paragraph*{Step 2: Discrete functions}
The discrete set of service time units for $\rho$, denoted by $T_\rho$, is the set obtained by dividing the service window $w_\rho$ into global intervals of length $\frac{1}{n}$, restricted to points with value at most $1$:
\[
T_\rho =
\left\{
t := \frac{1}{n} \cdot q
\;\middle|\;
 q \in \mathbb{N},\; t \in w_\rho,\; c_\rho(t) \le 1
\right\}.
\]
The new discrete penalty function is defined as $\hat{c}_\rho = c_\rho|_{T_\rho}$.

\begin{lemma} \label{lemma:disc:step2}
\begin{enumerate}
    \item For every request $\rho$ the set $T_\rho$ is not empty. 
    \item Let $T=\bigcup_\rho{T_\rho}$. Then, $|T| = O\left(nW\right)$.
    \item Any $\alpha$-approximation (competitive) algorithm for the problem using the new discrete penalty functions is a $6\alpha$-approximation (competitive) algorithm using the original non-zero penalty functions.
\end{enumerate}
\end{lemma}

\begin{proof}
\begin{enumerate}
\item Notice that the penalty functions obtained from Step~1 are still $1$-Lipschitz. Therefore, for every interval $I$ of length $\frac{1}{n}$, we have
\[
\max_{t \in I} c_\rho(t) - \min_{t \in I} c_\rho(t) \le \frac{1}{n}.
\]
Let $t_0$ be a time at which the original penalty function attains value $0$. The endpoint of the interval containing $t_0$ has cost at most $\frac{1}{n} \le 1$. By definition, this point belongs to $w_\rho$, and hence is included in $T_\rho$.

\item For every request, $|T_\rho| \le \frac{w_\rho}{1 / (n)} + 1$. Since all requests arrived in $[0,W]$, $|T| \le \frac{W}{1 / (n)} + 1 =  O\left(nW\right)$.

\item Since we only restrict the set of potential service points, $\text{ALG}_{\text{org}} \le \text{ALG}_{\text{new}}$.
Let $t$ be the service time of request $\rho$ in $\text{OPT}_{\text{org}}$. Define a new corresponding service time $t' \in T_\rho$ for $\text{OPT}_{\text{new}}$ as follows:
\begin{itemize}
    \item If $c_\rho(t) > \left(1-\frac{1}{n}\right)$, then $t'$ is an arbitrary point in $T_\rho$ (e.g., the closest one).
    \item Otherwise, $t'$ is the closest point in $T_\rho$ such that $t' \ge t$ (we will show that there exists one in this case).
\end{itemize}
Next we analyze the cost of the new optimal solution. Let $R_{> \left(1-\frac{1}{n}\right)}$ be the set of requests served by $\text{OPT}_{\text{org}}$ with penalty greater than $\left(1-\frac{1}{n}\right)$, and define $R_{\le \left(1-\frac{1}{n}\right)}$ analogously. Denote their total penalties by
$ cost\bigl(R_{> (1-\frac{1}{n})}\bigr)$
and $cost\bigl(R_{\le (1-\frac{1}{n})}\bigr)$.
Then, using these notations:
\[
\text{OPT}_{\text{org}}
=
cost\bigl(R_{> (1-\frac{1}{n})}\bigr)
+
cost\bigl(R_{\le (1-\frac{1}{n})}\bigr)
+
C,
\]
where $C$ is the total cost for performing services. We analyze two cases:
\paragraph*{Case 1 -- $R_{> \left(1-\frac{1}{n}\right)}$:}
Serving each request $\rho$ in this set individually costs at most $\hat{c}_\rho\left(t'\right) + c(P(v_\rho))$, where $v_\rho$ is the request's node and $P(v_\rho)$ is the path from $v_\rho$ to the root of the tree. Since $\hat{c}_\rho(t') \le 1$ as $t'\in T_\rho$, and the cost of that path is at most $2$ by the $2$-decreasing property, the total cost for serving $\rho$ is at most $3$.
In addition, since
\[
cost\bigl(R_{> (1-\frac{1}{n})}\bigr)
>
\left(1-\frac{1}{n}\right) \cdot \left|R_{> (1-\frac{1}{n})}\right|,
\]
then the total cost in $\text{OPT}_{\text{new}}$ for these requests (including the cost of performing the services and the cost of the penalties) is at most
\[
3 \cdot \left|R_{> (1-\frac{1}{n})}\right|
<
\frac{3}{1-\frac{1}{n}} \cdot cost\bigl(R_{> (1-\frac{1}{n})}\bigr)
\le
6 \cdot cost\bigl(R_{> (1-\frac{1}{n})}\bigr)
\]
for $n \ge 2$.

\paragraph*{Case 2 -- $R_{\le \left(1-\frac{1}{n}\right)}$:}
First, we show that in this case, $t'$ is the end of the global interval containing $t$. Indeed, since both $t$ and $t'$ lie in the same interval of length $\frac{1}{n}$, we have
\[
\hat{c}_\rho\left(t'\right) \le c_\rho(t) + \frac{1}{n} \le  \left(1-\frac{1}{n}\right) + \frac{1}{n} = 1.
\]
Therefore, $t' \in T_\rho$ as desired. 
Second, since $c_\rho(t) \ge \frac{1}{n}$ from the first step, we also conclude that $\hat{c}_\rho\left(t'\right) \le 2c_\rho(t)$.
Since $t' \ge t$, all requests in $R_{\le \left(1-\frac{1}{n}\right)}$ that were served at time $t$ in $\text{OPT}_{\text{org}}$ can be served together at time $t'$. This implies that these requests are served using the same number and structure of services as in $\text{OPT}_{\text{org}}$, and therefore no additional service cost is incurred for $R_{\le \left(1-\frac{1}{n}\right)}$.
Combining both cases, we get that 
\[
\text{OPT}_{\text{new}} 
\le 6 \cdot cost\left(R_{> \left(1-\frac{1}{n}\right)}\right) + 2\cdot cost\left(R_{\le \left(1-\frac{1}{n}\right)}\right) + C 
\le 6 \cdot \text{OPT}_{\text{org}}.
\]
\end{enumerate}
\end{proof}

\paragraph*{Maintaining the reductions in an online fashion} We assume that $n$ and $W$ are known in advance, and later explain how to get rid of this assumption, losing only a constant factor in the competitive ratio. We also take care of maintaining Lipschitz constant 1. The \textbf{first step} is applied independently to each request, without requiring any knowledge of other requests. Therefore, this step remains valid as new requests arrive online. The \textbf{second step}, which constructs a common discrete domain $T$, can also be maintained incrementally in the online setting. Specifically, we build $T$ in the order that requests arrive, by continuously adding new global intervals of length $1/n$, from each new request’s penalty function.

\subsection{Reduction to windows of length W} \label{subsec:window}
We can restrict attention to time intervals of length $W$ only, losing a constant factor in the competitive ratio. As a result, we may assume that the number of requests is $\bar{n}$, the maximum number of requests that arrive within any window of length~$W$. Typically, $\bar{n} \ll n$. 

To do so, we maintain two sequential, disjoint \emph{live} copies of the algorithm at any time, each spanning an interval of length $2W$ and overlapping by $W$. Specifically, every request arriving in the interval $[0, W)$ is assigned to the first copy, and every request arriving in $[W, 2W)$ is assigned to the second copy. Requests assigned to the first copy may be served up to time $2W$, while requests in the second copy may be served up to time $3W$.

More generally, for every $z \in \mathbb{Z}_{\ge 0}$, the even copy indexed by $2z$ contains all requests arriving in the interval $[2zW, (2z+1)W)$, and the odd copy indexed by $2z+1$ contains all requests arriving in $[(2z+1)W, (2z+2)W)$. Each copy is active for a duration of $2W$, and the algorithm is executed independently on each active copy, considering only the requests assigned to it.

Note that the feasibility of a solution is not violated by this partitioning method, as each request can still be served in its original service time. The total cost incurred by the algorithm is at most the sum of the costs over all copies. Since the windows of all even copies are disjoint, and likewise the windows of all odd copies are disjoint, the optimal solution must pay at least the maximum cost among the even copies or the odd copies. Consequently, this transformation increases the competitive ratio by at most a factor of~$2$.


\subsection{Eliminating the need to know n and W in advance}
\label{subsec:doubling}

Our algorithm requires an upper bound on the global discretization scale $\tfrac{1}{n}$ to extract the potential time units for service from each penalty function. Thus, we now explain why one may assume prior knowledge of the number of requests $n$, and the maximum length of a service window $W$, while incurring only a constant-factor loss in the competitive ratio. We also note that due to the normalization, we consider an update of the Lipschitz parameter by updating the parameter $W$.

\begin{theorem}
Let $ALG$ be an $O(\beta)$-competitive algorithm for MLAP under the assumption that the parameters $n$ and $W$ are known in advance, where
$
\beta \;=\; \log(n)\,\log\!\left( nDW  \right),
$
and where $D$ is fixed and known in advance. Then, there exists an algorithm $ALG'$ for MLAP that achieves the same asymptotic bound $O(\beta)$ without prior knowledge of $n$ and $W$.
\end{theorem}
\label{thm_knowing_parameters}


\begin{proof}
We apply a variant of the doubling technique that simultaneously handles all unknown parameters.
The algorithm dynamically maintains estimates of $n$ and $W$, updating them whenever they are violated. At any moment, the algorithm tracks the true values of these parameters, but uses them only to update the estimates at the beginning of a new phase.

Let $n_i$, and $W_i$ denote the estimates used in phase~$i$, and let $n(t)$,  and $W(t)$ be the true parameter values at time~$t$.

We note that the parameter $W$ is defined only after normalizing the penalty function via the following steps:
\begin{enumerate}
    \item Divide the function by the minimum positive edge weight incident to the root.
    \item Subtract the minimum value of the function from all values.
    \item Rescale the function to be $1$-Lipschitz. That is, if $f(x)$ is an $L$-Lipschitz, we instead use the function $f\!\left(\frac{x}{L}\right)$, and define $W$ accordingly.
\end{enumerate}

After the arrival of the first request, we set $W_1$ to the true values revealed by that request, and set $n_1 = 2$.
Whenever one of the estimates is violated, a new phase begins, the algorithm is restarted and the estimates are updated as follows:

\begin{itemize}
    \item If $n(t) > n_i$, update
    \[
        n_{i+1} \gets (n(t))^2.
    \]

    \item If $W(t) > W_i$, update
    \[
        W_{i+1} \;\gets\; n(t)\cdot D\cdot (W(t))^2.
    \]

\end{itemize}
Two key facts:
\begin{enumerate}
    \item For every request $\rho$, since the original arbitrary penalty function have values from $0$ to $1$ within its service window, it follows that $   w_\rho \ge\; 1 $ and hence at every time~$t$,
    \[
        W(t) 
        \;=\;
        \max_{\rho \le t} w_\rho 
        \;\ge\; 1.
    \]

    \item If several parameters are updated simultaneously, since the update rules use the \emph{current} value of each parameter, there is no circular dependency.
\end{enumerate}
To show that all updates are feasible, we use the following lemma: 
\begin{lemma}
For every phase $i \ge 1$:
\begin{enumerate}
    \item Each parameter strictly increases when updated:  
    \[
        x_{i+1} > x_i,
        \qquad x \in \{n,W\}.
    \]
    \item For every time $t$ during phase $i$, each estimate upper bounds the true value:
    \[
        n(t) \le n_i, \qquad W(t) \le W_i .
    \]
\end{enumerate}
\end{lemma}

\begin{proof}
In the first phase, $W_1$ is equal its true initial value and $n_1 = 2$ trivially upper bounds the single request observed so far. We check each possible update:

\begin{enumerate}
    \item If $n(t) > n_i$, then since $n_1 = 2$,  
    \[
        n_{i+1} = (n(t))^2 > n(t) > n_i.
    \]

    \item If $W(t) > W_i$, then using $W(t) \ge W_1$, $n(t) \ge 1$, $D\ge1$, and $W(t)  \ge 1$,
    \[
        W_{i+1}
        = n(t)\cdot D \cdot{(W(t))^2} \cdot 
        \;\ge\; W(t)
        \;>\; W_i.
    \]

\end{enumerate}
\end{proof}
Next we analyze the change in the competitive ratio after each update. The competitive ratio of the algorithm in phase $i$ is $\gamma \beta_i$, where $\gamma$ is some constant, $\beta_i= \log(n_i)\,\log\!\left(n_iD W_i \right)$, while using estimates $n_i$ and $W_i$.  
Then,
\begin{equation} \label{eq:competitive}    
ALG'
= \sum_i ALG_i
\;\le\;
\sum_i \gamma \beta_i \cdot OPT_i
\;\le\;
OPT \cdot \gamma\sum_i \beta_i.
\end{equation}
We now show that the competitiveness at least doubles in every phase, i.e., for every phase $i$, $\beta_{i+1} \;\ge\; 2\beta_i$. We analyze all three update cases.
\paragraph*{Case 1: Update triggered by $n$.}
Using $n(t) > n_i$,
\[
\begin{aligned}
\beta_{i+1}
&= \log((n(t))^2)\,
   \log\!\left( (n(t))^2 D W_i  \right)
\ge
2\log(n_i)\,
  \log\!\left( Dn_i  W_i  \right)
= 2\beta_i.
\end{aligned}
\]

\paragraph*{Case 2: Update triggered by $W$.}
Using $W(t) > W_i \ge W_1$,
\[
\begin{aligned}
\beta_{i+1}
&= \log(n_i)\,
   \log\!\left(
        n_iD \cdot n(t) \cdot D \cdot{(W(t))^2}       
   \right)
\ge
\log(n_i)\,
\log\!\left(
        n_i^2 D^2 W_i^2 
    \right)
= 2\beta_i.
\end{aligned}
\]

Let $\hat{n}$ and $\hat{W}$ be the final parameter values at the last phase, and let $\hat{\beta}$ denote the competitive ratio of the algorithm at the last phase.  
By the update rules,
\[
\hat{n} \le n^2,
\qquad
\hat{W} \le n D {W^2}.
\]
Thus,
\[
\begin{aligned}
\hat{\beta}
& =
\log\left(\hat{n}\right)\log\left(\hat{n}\hat{W}\right)
\le
\log(n^2)\,
\log\!\left(
       n^2 \cdot
       nD {W^2}
\right)
\\&\le
2\log(n)\,
 \log\!\left(
     n^3 D^3 W^3 
 \right)
=
6 \log(n)\,\log\!\left(nDW\right)
= 6\beta,
\end{aligned}
\]

Since $\beta_{i+1} \ge 2\beta_i$ and $\hat\beta \le 6\beta$,
\[
\sum_i \beta_i \;\le\; 12\,\beta.
\]
Combining with equation \ref{eq:competitive} completes the proof.
\end{proof}

%% file: Sections/4.reduction_imp.tex
\section{Reduction from MLAP to IMP}
\label{sec:reduction_imp}

\begin{figure}[t]
    \centering
    \includegraphics[width=0.85\linewidth]{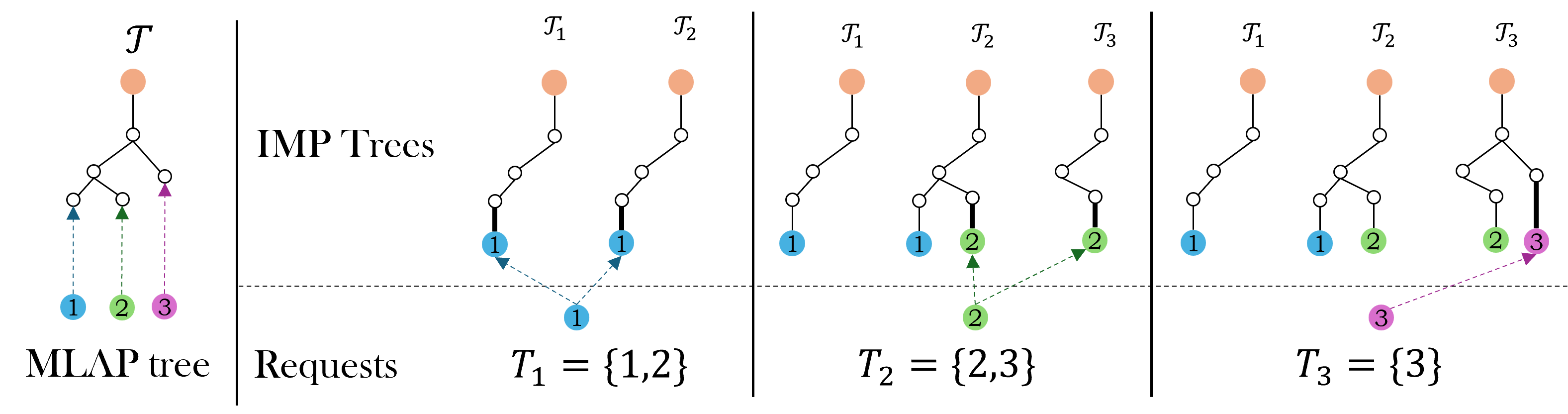}
    \caption{
    Each discrete time point induces a copy of the MLAP tree containing the path from the request’s node to the root, plus a new node and edge (in bold) encoding the penalty.
The first (blue) request has two service points, corresponding to trees $\mathcal{T}_1$ and $\mathcal{T}_2$. The second (green) request also has two service points: the first reuses an existing tree, adding only the missing part of the path. The third (purple) request has a single service point, corresponding to a new node added to an existing tree.
    }
    \label{fig:reduction_to_IMP}
\end{figure}

In this section, we present the third reduction, from MLAP with $2$-decreasing trees and discrete penalty functions to the Incremental Multicast Problem. This reduction is novel and establishes a connection between these two classic online problems.

\begin{theorem}
Let $ALG_{\text{IMP}}$ be an $\alpha$-competitive algorithm for $\text{IMP}(n,m)$, where $n$ is the number of requests and $m$ is the total number of edges in the resulting trees. Then, there exists an $\alpha$-competitive algorithm $ALG_{\text{MLAP}}$ for $\text{MLAP}_{\text{2-Discrete}}(n,D,T)$, where $D$ is the maximum depth of the trees in $\mathcal{T}$, $T$ is the final set of potential service points across all penalty functions and $n$ is the number of requests. Moreover, the number of edges satisfies $m = O(n D |T|)$.
\end{theorem}

The key idea behind this reduction is the analogy between potential service time units and facilities. A client’s connection cost to a facility is derived from the penalty cost of serving the corresponding request at that time unit. 
\begin{proof} 
\textbf{(1) Input MLAP $\Longrightarrow$ IMP:}
Let $(\mathcal{T}_{MLAP}, \mathcal{R}_{MLAP})$ be an instance of $\text{MLAP}_{\text{2-Discrete}}$, where $\mathcal{T}_{MLAP}$ is a $2$-decreasing weighted rooted tree and $\mathcal{R}_{MLAP}$ is a sequence of requests arriving over time. For simplicity, we assume that $\mathcal{T}_{MLAP}$ consists of a single $2$-decreasing tree rooted at node $r$, with edge cost function $c$. We denote by $c(v)$ the cost of the edge connecting node $v$ to its parent. The reduction naturally extends to a family of disjoint trees.  

Algorithm $ALG_{\text{IMP}}$ starts with an empty family of trees $\mathcal{T}_{IMP}$ and an empty sequence of requests $\mathcal{R}_{IMP}$. Remark that each request $\rho \in \mathcal{R}_{MLAP}$ is defined by the tuple $(v_\rho, a_\rho, c_\rho)$, where $v_\rho$ is a leaf node in $\mathcal{T}_{MLAP}$ associated with the request, $a_\rho$ is the arrival time of the request, and $c_\rho$ is a discrete penalty function defined over the set of potential service time units $T_\rho$. We denote by $P(v_\rho)$ the path from $v_\rho$ to the root $r$.

Each time unit $t \in T_\rho$ corresponds to a tree $\mathcal{T}_t' \in \mathcal{T}_{IMP}$ rooted at node $r_t$. In each such tree, the request is associated with a new leaf node $v_\rho^{(t)}$. This tree is a duplication of the original MLAP tree $\mathcal{T}$, restricted to the path $P(v_\rho)$ plus the new leaf node.  

Formally, for every request $\rho \in \mathcal{R}_{MLAP}$ (in the order of arrival), we create a corresponding request $\rho' \in \mathcal{R}_{IMP}$ defined as $(S_{\rho'}, \{\mathcal{P}_{v}\}_{v \in S_{\rho'}}, c_{\rho'})$
with the following components:
\begin{itemize}
    \item $S_{\rho'} = \{v_\rho^{(t)} \mid t \in T_\rho\}$.
    \item $\{\mathcal{P}_{v_\rho^{(t)}}\}_{v_\rho^{(t)} \in S_{\rho'}}$ is the collection of the incremental paths. Each path $\mathcal{P}_{v_\rho^{(t)}}$ is constructed by duplicating the path $P(v_\rho)$ from $\mathcal{T}$ into the tree $\mathcal{T}_t' \in \mathcal{T}_{IMP}$, taking only the portion of the path that does not already exist in that tree. Finally, a new leaf node $v_\rho^{(t)}$ is introduced, connected to the original node $v_\rho$.
    \item The cost function $c_{\rho'}$ on the new edges is inherited from the original edge costs in $\mathcal{T}$, except the only one edge connecting $v_\rho^{(t)}$ to $v_\rho$ which has the corresponding penalty cost in that time $c_\rho(t)$.
\end{itemize}
We denote by $T_{\rho'} = \{r^{(t)} \mid t \in T_\rho\}$ the set of all roots corresponding to each node in $S_{\rho'}$.  
Note that $T_\rho$ represents potential times to serve request $\rho$ in the MLAP, while $T_{\rho'}$ represents the corresponding roots in the constructed IMP trees. 
Although the notation may seem confusing at first, the precise correspondence between these sets  is the key to the reduction.  Figure~\ref{fig:reduction_to_IMP} illustrates the main idea of the reduction.

\textbf{(2) Apply algorithm $ALG_{\text{IMP}}$.} Upon the arrival of a request $\rho$, after converting it as described above, algorithm $ALG_{\text{MLAP}}$ simulates algorithm $ALG_{\text{IMP}}$ on the new request $\rho'$ using the demand $(S_{\rho'},T_{\rho'})$.  

\textbf{(3) Output IMP $\Longrightarrow$ MLAP:}
Let $E_{\rho'}$ denote the set of edges from $\mathcal{T}_{IMP}$ that $ALG_{\text{IMP}}$ adds to the solution to satisfy $\rho'$. $ALG_{\text{MLAP}}$ collects all these edges.  
At each time $t$ for which the IMP solution contains edges from a tree $\mathcal{T}_t'$, $ALG_{\text{MLAP}}$ performs a service. This service consists of the rooted subtree formed by those edges, and a request $\rho$ is included in this service only if its extra edge also appears in the IMP solution.  
In particular, if the IMP algorithm adds edges belonging to a tree corresponding to a past time unit, those edges are ignored in the MLAP solution. Therefore, the feasibility condition of making a service before time is expired, is preserved.  

In addition, each new request connects only to trees corresponding to its arrival time and future time units. Since the IMP algorithm satisfies each request immediately and never removes edges once added, all such edges remain in both solutions and collectively yield a feasible solution for the MLAP.

The cost of algorithm $ALG_{\text{MLAP}}$ is at most the cost of algorithm $ALG_{\text{IMP}}$, since every payment of $ALG_{\text{MLAP}}$ corresponds directly to a payment of $ALG_{\text{IMP}}$, either as part of the service cost (inner edges copied from $\mathcal{T}$) or as a penalty (the extra edge corresponding to the request). Note that some payments of $ALG_{\text{IMP}}$ may be ignored in $ALG_{\text{MLAP}}$.
From the opposite direction, using the reduction, any solution for the MLAP can be directly transformed into a valid solution for the corresponding IMP instance. Therefore, $OPT_{\text{IMP}} \le OPT_{\text{MLAP}}$, so,
$ALG_{\text{MLAP}} \le ALG_{\text{IMP}} \le \alpha \cdot OPT_{\text{IMP}} \le \alpha \cdot OPT_{\text{MLAP}}$.
Finally, the total number of edges in the resulting IMP instance is at most $O(nD|T|)$, since each request contributes at most $D+1$ edges across at most $|T|$ trees.  
\end{proof}

%% file: Sections/5.algorithm_imp.tex
\section{Algorithm for the online IMP} 
\label{sec:algorithm}
To complete the discussion on the MLAP with arbitrary penalty functions and to obtain a full algorithm for that problem, it remains to describe an algorithm for the IMP. In this section, we present a randomized algorithm for the online IMP, based on the randomized algorithm of Alon et al.\ for the online multicast problem \cite{Al:04}. The adapted algorithm yields the following theorem:
\begin{theorem}
\label{thm:IMP}
There exists a randomized $O\!\left(\log {n'} \log {m'} \right)$-competitive algorithm for the Incremental Multicast Problem in trees, where $n'$ is the current number of requests and $m'$ is the current number of edges in the trees.
\end{theorem}

Most of the algorithmic components and their analysis are naturally adapted from the work of Alon et al.~\cite{Al:04}. We present the algorithm construction together with all relevant lemmas and theorems. Since our setting is a generalization of the multicast problem, we also rewrite the proofs and include them in Appendix~\ref{appendix:algorithm}.

The most fundamental difference between the problems, due to the incremental nature of the IMP, concerns the knowledge of the structure of the graph in advance. The approach by Alon et al.~\cite{Al:04} relies on constructing a fractional feasible solution, which is subsequently rounded to yield a randomized integral solution. This algorithm uses the number of edges $m$ to normalize each edge's cost and weight (fraction). To address this, we apply the same doubling technique as mentioned in Subsection \ref{subsec:doubling}: we start with an estimate $m = 2$. Each time the estimation of $m$ is exceeded, we square its current real value. After every such an update, we start a new copy of the problem and simulate the algorithm from the beginning. We take to the solution all edges that were taken in all the copies. This technique introduces an additional constant factor to the performance compared with the setting where the final $m$ is known in advance. Notice that at each moment, the algorithm uses the current number of edges $m'$, and this parameter also reflected in the competitive ratio.

At a high level, the randomized algorithm for the IMP proceeds as follows:





\begin{algorithm}[!htbp]
\caption{\textbf{Randomized Algorithm for the IMP}
on $(S_\rho, \{\mathcal{P}_v\}_{v\in S_\rho}, c_\rho)$}
\label{alg:imp}

\small
\textbf{1. Preprocessing step.}
Remove every edge $e \in \{\mathcal{P}_v\}_{v\in S_\rho}$ with $c_\rho(e) > \alpha$, together with all edges in the path downwards until the request node.
Add to the solution all edges with $c_\rho(e) \le \frac{\alpha}{m}$.
Normalize the cost of the remaining edges to lie in $[1,m]$.

\medskip
\textbf{2. Graph increment.}
Add to the current trees $\mathcal{T}$ all new nodes and edges in $\{\mathcal{P}_v\}_{v\in S_\rho}$. Each new edge is assigned an initial weight (fraction) of $\frac{1}{m}$.

\medskip
\textbf{3. Fractional solution.}
Perform a \emph{weight augmentation step} for the demand $(S_\rho, T_\rho)$
as described in Algorithm~\ref{alg:weight_augmentation}.

\medskip
\textbf{4. Integral solution.}
Perform a \emph{rounding step} on every tree rooted at a node in $T_\rho$
as described in Algorithm~\ref{alg:rounding}.

\end{algorithm}

\subsection{Fractional Solution}
\label{subsec:fractional}
We present a fractional algorithm that is $O(\log m')$-competitive, where $m$ is the number of edges in all trees in $\mathcal{T}$. 
To satisfy a demand $(S_\rho,T_\rho)$, the algorithm performs the \emph{weight augmentation} step described in Algorithm~\ref{alg:weight_augmentation}. We denote the cost of an edge by $c_e$ and the weight of an edge by $w_e$.


\begin{algorithm}
\caption{\textbf{Weight Augmentation Step} for a demand $(S_\rho, T_\rho)$}\label{alg:weight_augmentation}
\While{the maximum flow from $S_\rho$ to $T_\rho$ is less than $1$}{
    \For{each edge $e$ in a minimum cut between $S_\rho$ and $T_\rho$}{
      $w_{e} \leftarrow w_{e} \cdot \left(1 + \frac{1}{c_e}\right)$
    }
}
\end{algorithm}

The performance of the fractional solution is summarized in the following theorem:
\begin{theorem}
\label{thm:frac}
The fractional algorithm satisfies all requests and $O\!\left(\log m\right)$-competitive for the fractional IMP.
\end{theorem}

To prove this theorem, we use the following lemma:



\begin{lemma}
\label{lem:frac2}
The total number of weight-augmentation steps performed by the algorithm is at most
$\alpha + \alpha \log m = O\!\left(\alpha \log m\right)$,
where $\alpha$ is the cost of an optimal solution, and $m$ is the number of edges in $\mathcal{T}$.
\end{lemma}

This lemma is the heart of the analysis of the fractional solution. It is based on the following potential function:
$
\Phi = \sum_{e \in E} c_e w_e^* \log w_e,
$
where $w_e^*$ is the weight of edge $e$ in the optimal solution. Although the full graph is not known in advance, for analysis purposes the lemma considers the final set of edges $E$ retrospectively. If an edge does not exist yet in the graph at a given moment, its weight is treated as the initial weight.

Combining the last lemma with the simple observation that any weight augmentation step increases the value of the algorithm by at most 1, gives us the desired result. Full proofs are provided in Appendix~\ref{appendix:fractional}.

\subsection{Randomized Solution}
\label{subsec:rounding}
We now introduce the \emph{rounding step}, which is performed after each weight augmentation step. For every tree $T$ associated with a request, the algorithm executes the rounding step on $T$. Specifically, a random threshold is sampled for that tree, and this threshold determines which of its edges are rounded (i.e., included in the solution). The procedure is described in Algorithm~\ref{alg:rounding}.

\begin{algorithm}[!htbp]
\caption{\textbf{Rounding Step} on a Tree $T \in \mathcal{T}$}
\label{alg:rounding}

\small
\textbf{1. Random thresholds.}
Sample $q = 2\left\lceil \log(n' + 1) \right\rceil$ independent random variables $X(T,j)$ for $1 \le j \le q$, each drawn uniformly from $[0,1]$. Define the threshold of $T$ as $\theta(T) := \min_{1 \le j \le q} X(T,j)$.

\medskip
\textbf{2. Edge rounding.}
Include in the solution every edge $e \in T$ whose fractional weight satisfies $w_e > \theta(T)$.
\end{algorithm}

Let $\alpha$ denote the value of the optimal solution, $m'$ the current number of edges in the trees, and $n'$ the current number of requests. Note that the algorithm does not require knowledge of the final number of requests. Instead, it uses the current value of $n'$, and as new requests arrive, more random variables are added to each tree. Consequently, the threshold of each tree can only decrease over time, which means that additional edges may be added to the solution, but once an edge is included, it remains in the solution.

\begin{lemma}
\label{lem:rand}
At any moment during the run of the algorithm, the following hold:
\begin{enumerate}
    \item The expected cost of the algorithm is $O(\alpha \log{n'} \log{m'})$.
    \item For any request $\rho$, the probability that $\rho$ is not served is at most $\frac{1}{{n'}^2}$.
\end{enumerate}
\end{lemma}

The lemma shows that the randomized IMP algorithm yields a feasible solution with high probability while preserving a good competitive ratio. To ensure feasibility, if a request $\rho$ is not served, the algorithm serves it using the tree with minimum path cost, adding at most $OPT$. Since this occurs with probability at most $\frac{1}{{n'}^2}$, the expected extra cost is negligible and does not affect the asymptotic bound. The full proof appears in Appendix~\ref{appendix:rounding}. Overall, this gives an $O(\log n' \log m')$-competitive randomized algorithm for IMP.


%% file: Sections/6.lower_bound.tex
\section{Lower Bound}
\label{sec:lower_bound}

We establish a lower bound of $4$ on the competitive ratio for any deterministic algorithm, even in the special case where $D = 1$ (i.e., the TCP Acknowledgment Problem) and even when all penalty functions have only the values $0$ or $1$. One can show the same result with continuous penalty function that gets values in $[0,1]$. 
Moseley et al. in ~\cite{Mo:24} and Shmoys et al.~\cite{S:25} designed deterministic algorithms for the single-item JRP (equivalent to the TCP-AP) under the assumption that the penalty functions initially decrease and then increase, that are less than $4$ competitive. Therefore, our lower bound shows that allowing \emph{arbitrary} penalty functions, even with only two possible values, makes the problem strictly harder.

\begin{theorem} \label{theorem:lower_bound}
Every deterministic algorithm for the TCP-AP with arbitrary penalty functions, even when each function has only the values $0$ or $1$, has competitive ratio at least $4-\epsilon$ for every $\epsilon>0$.
\end{theorem}
 
The proof relies on the following arithmetic lemma.
\begin{lemma} \label{Lemma:lower_bound}
For every $\epsilon > 0$ and for every sequence of real numbers $\{c_k\}_{k=1}^\infty$ (for consistency we define $c_0=0$) with $c_k \ge 1$ for all $k \ge 1$, there exists an $n\in \mathbb{N}$ such that
\[
\frac{s_n}{\,c_{n-1} + 1\,} > 4-\epsilon,
\]
where $s_n := \sum_{k=1}^n c_k$ denote the partial sums of the sequence.
\end{lemma}

\begin{proof}
Let $\epsilon > 0$ and let $\{c_k\}_{k=1}^\infty$ be a sequence of real numbers with $c_k \ge 1$ for all $k \ge 1$. Since the left-hand side of the desired inequality is always positive, we can assume that $\epsilon < 4$.
Assume toward a contradiction that 
\[
\frac{s_n}{c_{n-1}+1} \le 4-\epsilon, \qquad \text{for every } n.
\]
By denote $y_n := \frac{s_{n+1}}{s_n}$, we can write,
\[
y_{n+1}
    = \frac{s_{n+2}}{s_{n+1}}
    \le \frac{(4-\epsilon) \,(c_{n+1}+1)}{s_{n+1}}
    = \frac{(4-\epsilon) \,(s_{n+1} - s_n + 1)}{s_{n+1}}
    = (4-\epsilon)\left(1 - \frac{1}{y_n} + \frac{1}{s_{n+1}}\right).
\]
Therefore,
\begin{align*}
(y_{n+1} - y_n)\,y_n
    =& -y_n^2 + y_n y_{n+1} \\
    \le& -y_n^2
        + (4-\epsilon) y_n \left(1 - \frac{1}{y_n} + \frac{1}{s_{n+1}}\right) \\
    =& -y_n^2 + (4-\epsilon) y_n - (4-\epsilon) + \frac{4-\epsilon}{s_{n}}.
\end{align*}
Since $s_{n} \ge n$ (as $c_n\ge1$), we further obtain that
\begin{equation} \label{ineq}
(y_{n+1} - y_n)\,y_n
    \le -y_n^2 + (4-\epsilon) y_n + (4-\epsilon)\left(\frac{1}{n} - 1\right).
\end{equation}
Consider the quadratic function $f(x) = -x^2 + (4-\epsilon) x + (4-\epsilon)\left(\frac{1}{n} - 1\right)$. Its discriminant is
\[
\Delta
    = (4-\epsilon)^2 + 4(4-\epsilon)\left(\frac{1}{n} - 1\right)
    = (4-\epsilon)^2 - 4(4-\epsilon) + \frac{4(4-\epsilon)}{n}.
\]
Note that for all sufficiently large $n$, (e.g. $n > \frac{4}{\epsilon}$), we have $\Delta < 0$.
Since the coefficient of $x^2$ is negative, $f(x) < 0$ for all $x$. 
Hence, for all such values of $n$, $(y_{n+1} - y_n)\,y_n \le f(y_n) < 0$. 
Since $y_n$ is always positive we get that $y_{n+1} < y_n$.

In addition, since every element in $c_n$ is positive, $s_{n+1} \ge s_n$, and therefore $y_n \ge 1$ for every $n$.
Thus, $y_n$ is eventually decreasing and bounded below by $1$, so it converges to some limit $y \ge 1$.  
By taking limits in both sides of inequality \ref{ineq} we get that $0 \le -y^2 + (4-\epsilon) y - (4-\epsilon)$.
However, the expression in the right-hand side is strictly negative for all $y$ if $\epsilon < 4$. This contradiction completes the proof.
\end{proof}

We now prove Theorem~\ref{theorem:lower_bound}.
\begin{figure}[t]
    \centering
    \includegraphics[width=0.75\linewidth]{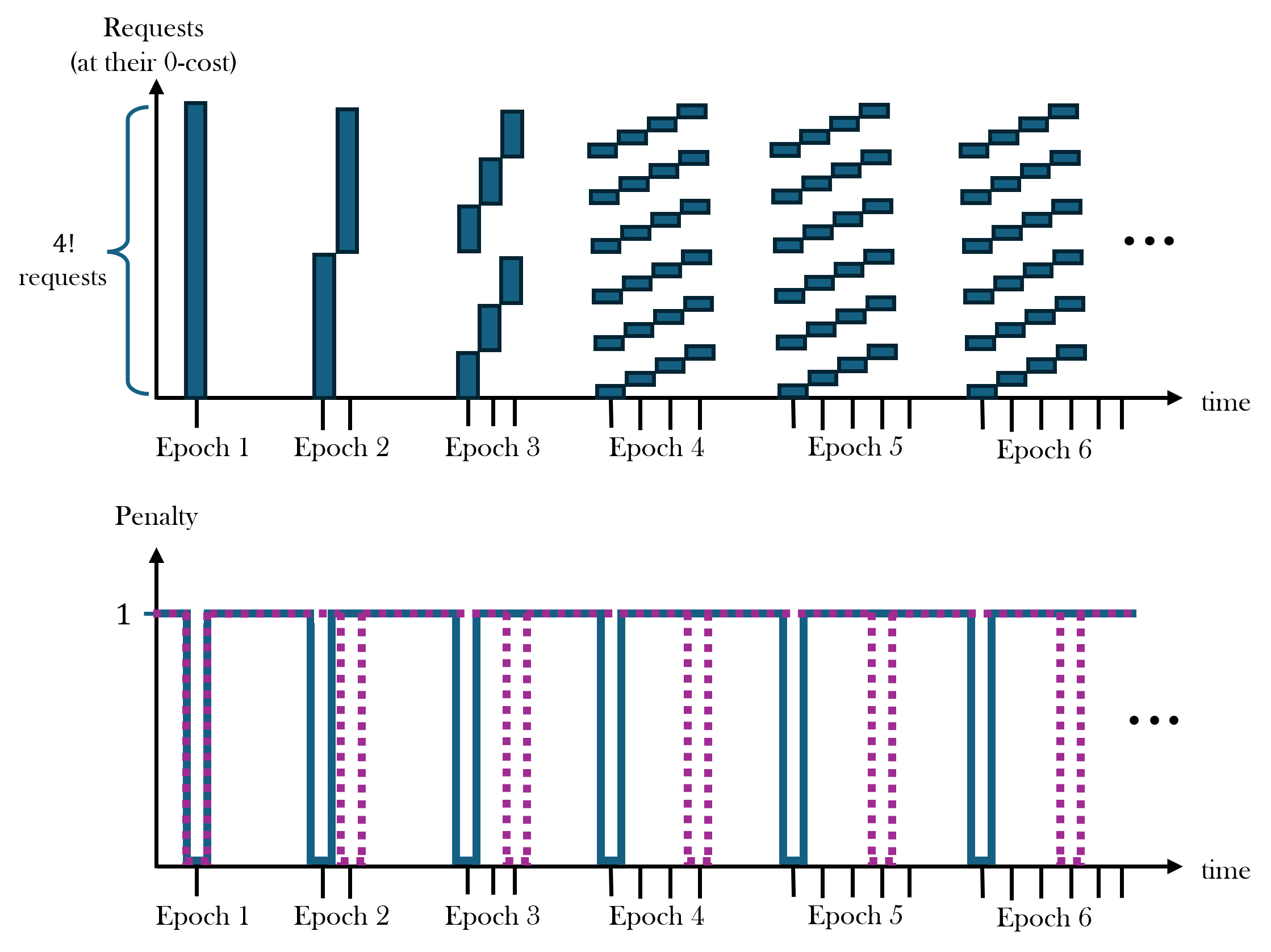}
    \caption{This image illustrates the release of the first group. 
    (a) The first figure shows the assignment of zero-penalty time units to requests in the first group across successive epochs. Each vertical column represents a group of requests that share the same zero-penalty time unit (all other times incur a penalty of~1). In each new epoch, every subgroup is evenly split across all time units of that epoch, assigning each part a distinct zero-penalty time.
    (b) The second figure shows the penalty functions of two representative requests from the first group. The solid blue line represents the penalty function of the first request in the first group, while the dashed purple line represents the penalty function of the last request of the first group (request~24).}
    \label{fig:Releasing the first group}
\end{figure}

\begin{figure}[t]
    \centering
    \includegraphics[height=0.4\textheight]{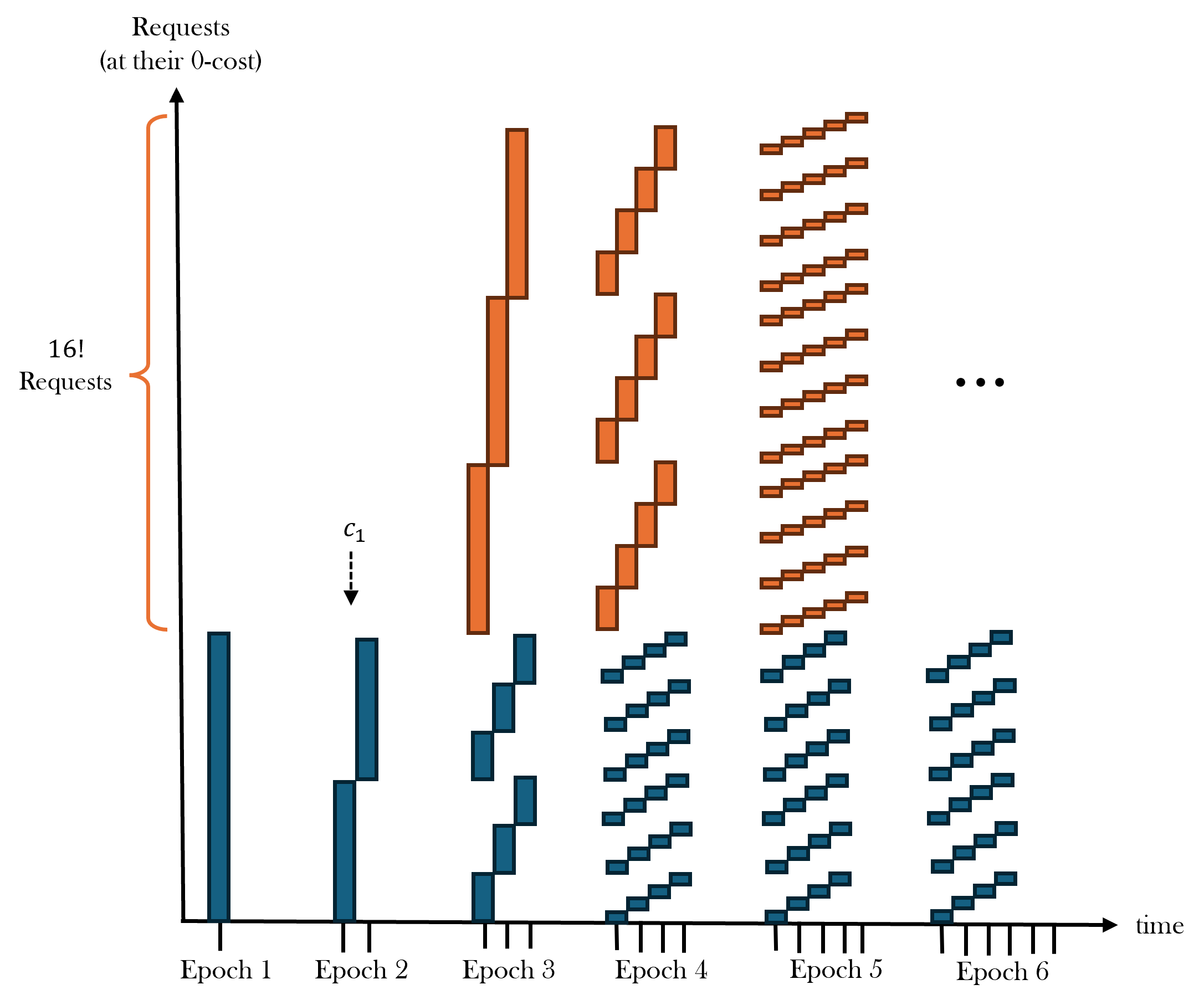}
    \caption{Once the algorithm selects the epoch in which the first group is being served ($c_1$), the second group is released with $(4^2)!=16!$ requests The requests zero-penalty time units are distributed across the next epochs using the same splitting procedure.}
    \label{fig:lower-bound-img}
\end{figure}

\begin{proof}[Proof of Theorem~\ref{theorem:lower_bound}]
Let ALG be an online algorithm for the problem and let $\epsilon>0$.  
The service cost at any time is equal to~$1$, i.e., we are in the setting of the TCP-AP.
Assume, toward a contradiction, that ALG is $(4-\epsilon)$-competitive.
We partition the timeline into disjoint consecutive epochs of increasing length:
epoch~$1$ consists of one time unit; epoch~$2$ consists of two time units; epoch~$3$ consists of three time units; and in general, epoch~$i$ consists of $i$ consecutive time units. All epochs are disjoint.

Requests are released in groups. The $k$-th group consists of $(a_k)!$ requests, where $a_k = 4^k$. Initially, group~$1$ is released, containing $(a_1)! = 4! = 24$ requests. All requests of this group have zero penalty at the single time unit of epoch~$1$. In epoch~$2$, half of the group (12 requests) has zero penalty at the first time unit, and the other half at the second time unit. In epoch~$3$, each half is divided into three equal subgroups, each containing four requests, with zero penalty at a distinct time unit of that epoch. From epoch~$4$ onward, each subgroup from the previous epoch is further subdivided so that each resulting subgroup has a distinct zero-penalty service time. Figure~\ref{fig:Releasing the first group} illustrates this construction.

Once ALG serves the last request of group~$k$, say in epoch~$c_k$, group~$k+1$ is released starting from epoch~$c_{k}+1$. The zero-penalty times for the new group are defined analogously, beginning at epoch~$c_{k}+1$. Specifically, upon its release, group~$k\!+\!1$ is first partitioned into $c_k+1$ subgroups, each assigned a distinct zero-penalty time unit within epoch~$c_k+1$; in each subsequent epoch~$c_k+\ell$, every existing subgroup from previous epoch $c_k+\ell-1$ is further partitioned into $c_k+\ell$ subgroups, with each assigned a different zero-penalty time unit of that epoch. Figure~\ref{fig:lower-bound-img} illustrates the overall release pattern. For convenience, we define $c_0 = 0$.

Note that for every $k \ge 1$, having $(a_k)!$ requests in the $k$-th group guarantees that all subgroups contain an integer number of requests throughout the construction, across all epochs from epoch~$c_{k-1}+1$, when the group is released, until the final epoch in which the group is fully partitioned.

We now make several simplifying assumptions about ALG, all of which can be enforced without increasing its cost:
\begin{enumerate}
    \item \textbf{Every request is served at a zero-penalty time.}  
    If ALG serves a request at a time with penalty~$1$, shifting the service to the nearest zero-penalty time does not increase the cost.

    \item \textbf{Whenever ALG performs a service, it serves all requests whose penalty is zero at that time.}  
    Serving such requests later cannot reduce the cost.

    \item \textbf{Each group is fully served within a single epoch.}  
    By the previous observation, all subgroups that share a zero-penalty time are served together.
    Suppose ALG serves some subgroups of a group in one epoch and the remaining subgroups in a later epoch.
    By construction, when serving the later subgroups, ALG can simultaneously serve all  earlier subgroups as well.
    Hence, serving the entire group in the later epoch can yield only a smaller total cost.
\end{enumerate}
Let $ALG(k)$ and $OPT(k)$ denote the total cost incurred by ALG and OPT, respectively, for serving all requests up to and including group~$k$.
To bound the values of ALG and OPT, consider the following lemma:
\begin{lemma}
For every $k \ge 1$, it holds that $c_k \le a_k$.
\end{lemma}
\begin{proof}
Since ALG is assumed to be $(4-\epsilon)$-competitive, for all $k \ge 2$,
\[
c_{k-1} + c_k \le ALG(k) \le (4-\epsilon)\,OPT(k) < 4(c_{k-1}+1).
\]
Rearranging yields
\[
c_k < 3c_{k-1} + 4 .
\]
Since $OPT(1)=1$, we obtain $c_1 < 4$. By induction, it follows that $c_k \le 4^k = a_k$ for all $k \ge 1$.
\end{proof}
Following this lemma, when ALG serves group~$k$ in epoch~$c_k$, it uses all time units of that epoch, and therefore incurs an additional cost of exactly~$c_k$. Hence,
\[
ALG(k) = \sum_{i=1}^k c_i .
\]
By construction, all requests from previous groups continue to have zero-penalty at the first time units of every subsequent epoch. Therefore, OPT can serve all requests up to group~$k$ in the first epoch in which group~$k$ is released, namely epoch~$c_{k-1}+1$. Thus,
\[
OPT(k) \le c_{k-1} + 1 .
\]
For each $k\ge 1 $, we consider the competitive ratio after serving all requests up to group~$k$:
\[
\frac{ALG}{OPT}(k) \ge \frac{\sum_{i=1}^k c_i}{c_{k-1}+1}.
\]
Since $c_k \ge 1$ for all $k \ge 1$, the sequence $\{c_k\}_{k \ge 1}$ satisfies the conditions of Lemma~\ref{Lemma:lower_bound}. Therefore, there exists $n$ such that
\[
\frac{ALG}{OPT}(n) \ge 4 - \epsilon.
\]
Hence, keeping the above construction, there exists group $n$ in which ALG is more than $(4-\epsilon)$-competitive. This is a contradiction.
\end{proof}

%% file: Sections/7.concluding_remarks.tex
\section{Conclusions and Open Problems}
\label{sec:conculding-remarks}
In this work, we initiate the study of the online Multi-Level Aggregation Problem (MLAP) with \emph{arbitrary penalty functions}, which may increase and decrease multiple times. This setting goes beyond the standard assumption of monotone non-decreasing (delay) functions. To the best of our knowledge, this is the first work to analyze such general penalty functions in the context of MLAP. A central contribution is a discretization of continuous-time penalty functions, which enables a reduction to a new generalization of the Multicast Problem that we introduce, the \emph{Incremental Multicast Problem} (IMP).

Several open questions remain. In particular, there is a significant gap between the competitive ratio of our algorithm and the lower bound we establish. Closing this gap, either by designing improved algorithms or by strengthening the lower bounds, remains an interesting direction, even for the special case $D=1$, corresponding to the TCP Acknowledgment Problem. More broadly, our discretization-and-reduction framework suggests that similar techniques may be applicable to other online problems with arbitrary penalty functions, potentially leading to a general and clean methodology for handling such settings.

%% file: Sections/8.apendix.tex
\newpage
\appendix

\section{Randomized algorithm for the IMP: Section \ref{sec:algorithm}}
\label{appendix:algorithm}
\subsection{Preprocessing step} 
\label{appendix:preprocess}
The original algorithm assumes that the cost of the optimal fractional solution, denoted by $\alpha$, is known in advance. As in the case of assuming knowledge of the number of edges $m$, we can apply a doubling technique and assume that $\alpha$ is known within a factor of $2$. After the first request arrives, we initialize $\alpha$ to be the smallest cost among all new edges introduced by that request. Whenever this estimate is violated, we double $\alpha$. This increases the overall cost of the solution by at most a constant factor of $2$.

After any violation of the estimates of either $\alpha$ or $m$, we discard the current weights assigned to the edges, update our estimate accordingly, and resume the algorithm. Note that the online setting does not allow edge weights to decrease. Thus, when we say we "forget" past weights, we mean that in each iteration, the algorithm only uses the weights assigned during that round. However, at any time, the actual weight of an edge is the maximum weight it has ever been assigned in any previous iteration.

Given that the number of edges $m$ and the optimal value $\alpha$ are known, we may assume that all edge costs lie between $1$ and $m$. To see it, we use the following observation:
\begin{observation}
\label{obs:weights}
In any fractional optimal solution, no edge with cost larger than $\alpha$ receives positive weight.
\end{observation}

\begin{proof} 
Suppose otherwise, let edge $e$ have $w(e) > 0$ while $c(e) > \alpha$. Because the total cost is at most $\alpha$, it must hold that $w(e) < 1$. Now, define a modified fractional solution $w'$ by setting $w'(e) = 0$ and scaling all other weights as  
\[
w'(e') = \frac{w(e')}{1 - w(e)}, \quad e' \neq e.
\]  
This yields a feasible solution of cost
\[
\frac{\alpha - c(e)w(e)}{1 - w(e)} < \frac{\alpha (1 - w(e))}{1 - w(e)} = \alpha,
\]
contradicting the optimality of OPT.
\end{proof}

Therefore, all edges with cost greater than $\alpha$ can be safely ignored by removing them and all nodes connected below them in the tree. Conversely, all edges with cost less than $\alpha/m$ may be included automatically in the solution, increasing the solution cost by at most an additive $\alpha$. Thus, after this preprocessing step, the remaining edges have costs in the range $[\alpha/m, \alpha]$. Finally, by multiplying with $m/\alpha$, we normalize these costs to lie between $1$ and $m$.

\subsection{Fractional solution: Full proofs from Subsection \ref{subsec:fractional}}
\label{appendix:fractional}

\begin{proof} [proof of Lemma \ref{lem:frac2}]
We define the following potential function:
\[
\Phi = \sum_{e \in E} c_e w_e^* \log(w_e),
\]
where $E$ is the set of edges in the graph, and $w_e^*$ is the weight assigned to edge $e$ in the optimal (fractional) solution. We show 3 properties of $\Phi$:
\begin{enumerate}
    \item \textbf{Initial value:} Initially, $w_e = \frac{1}{m}$ for every $e$, therefore:
    \[
    \Phi_0 = \sum_{e \in E} c_e w_e^* \log\left(\frac{1}{m}\right) \ge -\alpha \log m.
    \]

    \item \textbf{Upper bound:} During the execution of the algorithm, no weight $w_e$ exceeds $1 + 1/c_e$, since otherwise it would not be a part of an augmentation step. Since $c_e \ge 1$, for every edge, we have $w_e \le 2$. Thus,
    \[
    \Phi \le \sum_{e \in E} c_e w_e^* \log(2) \le \alpha
    \]
    \item \textbf{Augmentation increases $\Phi$ by at least 1:} Let $\mathcal{C}$ be a minimal cut where augmentation occurs. Then,
   \begin{align*}
    \Delta\Phi 
    &= \sum_{e\in\mathcal{C}}{c_e w_e^* \log\left(w_e\left(1+\frac{1}{c_e}\right)\right)}
    -
    \sum_{e\in\mathcal{C}}{c_e w_e^* \log(w_e)} \\
    & =
    \sum_{e\in\mathcal{C}}{c_e w_e^* \log\left(1+\frac{1}{c_e}\right)} \\
    &\ge 
    \sum_{e\in\mathcal{C}}{w_e^*} 
    \ge  1,
    \end{align*}
    where the last two inequalities holds since (1) $x\log(1+\frac{1}{x})$ is increasing function where $x\ge 1$, and (2) the optimal solution satisfy the request so every cut has total weight of at least 1.
\end{enumerate}
From properties (1) and (2), the value of $\Phi$ increases by at most $\alpha+\alpha\log m$ during the run of the algorithm. From property (3), each augmentation increases $\Phi$ by at least 1. Therefore, the total number of augmentation steps during the run of the algorithm is at most $\alpha+\alpha\log m = O(\alpha \log m)$.
\end{proof}

\begin{proof} [proof of Theorem \ref{thm:frac}]
We show that during the execution of the algorithm, the total cost incurred by the algorithm satisfies:
\[
ALG \le \alpha \log m + 2\alpha = O(\alpha\log m)
\]
We analyze the change in the algorithm's cost throughout the different steps performed during its execution. Specifically, we consider the impact of each of the following operations:
\begin{itemize}
    \item \textbf{Initial Condition:} At the beginning, every edge has weight $1/m$. Therefore, $ALG \le 1 \le \alpha$. 
    \item \textbf{Augmentation Step:} Let $\mathcal{C}$ be the set of edges involved in a weight augmentation step (a minimal cut). These edges have a total weight strictly less than 1. In the augmentation, each edge $e$ in $\mathcal{C}$ is increased by $\frac{w_e}{c_e}$, incurring a cost of $c_e \cdot \frac{w_e}{c_e} = w_e$. Hence, the increase in cost during this step is:
    \[
    \sum_{e \in \mathcal{C}} w_e < 1.
    \]
    From the previous lemma, we know that the number of augmentation steps is at most $\alpha \log m + \alpha$. Therefore, the total increase in $ALG$ due to augmentations is at most $\alpha \log m + \alpha$.
\end{itemize}
Overall, the value of ALG is at most $\alpha \log m + 2\alpha = O(\alpha\log m)$.
\end{proof}

\subsection{Rounding step: Full proof from Subsection \ref{subsec:rounding}}
\label{appendix:rounding}
\begin{proof} [proof of Lemma \ref{lem:rand}]
For the first part, for an edge $e$, define $Y(e,j) := [{w}_e > X(T,j)]$. Since each random variable is chosen uniformly in $[0,1]$, the probability and the expectation of the indicator $Y(e,j)$ is at most $w_e$. Therefore,
\begin{align*}
ALG &= \mathbb{E} \left[ \sum_{e\in E}{c_e} \right] \le \sum_{e\in E}{\sum_{j=1}^q}{c_e \cdot  \mathbb{E} \left[ Y(e,j) \right]}  \\
& \le \sum_{e\in E}{\sum_{j=1}^q}{c_e  w_e} = q\sum_{e\in E}{w_e c_e} \\
& \le q \cdot O(\alpha\log m)=O(\alpha \log {n'} \log m),
\end{align*}
where the last inequality uses the result of the fractional solution.
For the second part, fix a request $\rho$ and let $T_\rho$ be the roots of the trees that associated with $S_\rho$ (the nodes of $\rho$). From the feasibility of the fractional solution, we know that the total flow from $S_\rho$ to $ T_\rho$ is at least 1. For every $v\in S_\rho$, let $f_\rho^{T_v}$ be the total flow from $v$ to its corresponded root in the tree $T_v$. So, we can write:  $\sum_{v \in S_\rho} {f_\rho^{T_v}} \ge 1$. 
The probability that request $\rho$ is not satisfied using the node $v$ is at most the probability that $f_\rho^{T_v} < \theta(T_v)$. For any $1 \le j \le q$, the probability that the value $f_\rho^{T_v}$ is at most the value of $X(T,j)$ is $1-f_\rho^{T_v}$. Since all random variables are independent, the probability that $\rho$ is not satisfied by any of its associated trees due to the $j$'th random variable is at most
\[
\prod_{v \in S_\rho} (1 - f_\rho^{T_v}) \le e^{- \sum_{v \in S_\rho} {f_\rho^{T_v}}} \le \frac{1}{e}.
\]
Therefore, after $q = 2\lceil\log(n'+1)\rceil$ repetitions of independent random variables for each tree, the probability that $\rho$ is not served is at most
\[
\left(\frac{1}{e}\right)^q \le \frac{1}{n'^2}.
\]
\end{proof}

%% file: Sections/bibliography.bib
@article{Bu:09,
	author = {Buchbinder N. and Naor J.},
	year = "2007",
	title = {The Design of Competitive Online Algorithms
via a Primal–Dual Approach},
	journal = {Foundations and Trends in Theoretical Computer Science},
	volume = "3",
	number = "2–3",
	pages = "93-263",
}

@article{Al:04,
	author       = {Noga Alon and
                  Baruch Awerbuch and
                  Yossi Azar and
                  Niv Buchbinder and
                  Joseph Naor},
  title        = {A general approach to online network optimization problems},
  journal      = {{ACM} Trans. Algorithms},
  volume       = {2},
  number       = {4},
  pages        = {640--660},
  year         = {2006},
  doi          = {10.1145/1198513.1198522},
}

@article{Bi:15,
  author       = {Marcin Bienkowski and
                  Martin B{\"{o}}hm and
                  Jaroslaw Byrka and
                  Marek Chrobak and
                  Christoph D{\"{u}}rr and
                  Luk{\'{a}}s Folwarczn{\'{y}} and
                  Lukasz Jez and
                  Jir{\'{\i}} Sgall and
                  Kim Thang Nguyen and
                  Pavel Vesel{\'{y}}},
  title        = {Online Algorithms for Multilevel Aggregation},
  journal      = {Oper. Res.},
  volume       = {68},
  number       = {1},
  pages        = {214--232},
  year         = {2020}
}

@inproceedings{AT:19,
  author       = {Yossi Azar and
                  Noam Touitou},
  title        = {General Framework for Metric Optimization Problems with Delay or with
                  Deadlines},
  booktitle    = {60th {IEEE} Annual Symposium on Foundations of Computer Science, {FOCS}
                  2019, Baltimore, Maryland, USA, November 9-12, 2019},
  pages        = {60--71},
  publisher    = {{IEEE} Computer Society},
  year         = {2019},
  doi          = {10.1109/FOCS.2019.00013},
}

@inproceedings{Bu:17,
  author       = {Niv Buchbinder and
                  Moran Feldman and
                  Joseph (Seffi) Naor and
                  Ohad Talmon},
  editor       = {Philip N. Klein},
  title        = {\emph{O}(depth)-Competitive Algorithm for Online Multi-level Aggregation},
  booktitle    = {Proceedings of the Twenty-Eighth Annual {ACM-SIAM} Symposium on Discrete
                  Algorithms, {SODA} 2017, Barcelona, Spain, Hotel Porta Fira, January
                  16-19},
  pages        = {1235--1244},
  publisher    = {{SIAM}},
  year         = {2017},
  doi          = {10.1137/1.9781611974782.80},
}

@inproceedings{Mo:24,
  author       = {Benjamin Moseley and
                  Aidin Niaparast and
                  R. Ravi},
  title        = {Putting Off the Catching Up: Online Joint Replenishment Problem with
                  Holding and Backlog Costs},
  booktitle    = {Proceedings of the 2025 Annual {ACM-SIAM} Symposium on Discrete Algorithms,
                  {SODA} 2025, New Orleans, LA, USA, January 12-15, 2025},
  pages        = {3865--3883},
  publisher    = {{SIAM}},
  year         = {2025},
  doi          = {10.1137/1.9781611978322.130}
}

@inproceedings{Bi:21,
  author       = {Marcin Bienkowski and
                  Bj{\"{o}}rn Feldkord and
                  Pawel Schmidt},
  editor       = {Markus Bl{\"{a}}ser and
                  Benjamin Monmege},
  title        = {A Nearly Optimal Deterministic Online Algorithm for Non-Metric Facility Location},
  booktitle    = {38th International Symposium on Theoretical Aspects of Computer Science, {STACS} 2021, March 16-19, 2021, Saarbr{\"{u}}cken, Germany (Virtual
                  Conference)},
  series       = {LIPIcs},
  volume       = {187},
  pages        = {14:1--14:17},
  publisher    = {Schloss Dagstuhl - Leibniz-Zentrum f{\"{u}}r Informatik},
  year         = {2021},
  doi          = {10.4230/LIPICS.STACS.2021.14},
}

@article{Bu:13,
  author       = {Niv Buchbinder and
                  Tracy Kimbrel and
                  Retsef Levi and
                  Konstantin Makarychev and
                  Maxim Sviridenko},
  title        = {Online Make-to-Order Joint Replenishment Model: Primal-Dual Competitive
                  Algorithms},
  journal      = {Oper. Res.},
  volume       = {61},
  number       = {4},
  pages        = {1014--1029},
  year         = {2013},
  doi          = {10.1287/OPRE.2013.1188},
}

@inproceedings{Bi:14,
  author       = {Marcin Bienkowski and
                  Jaroslaw Byrka and
                  Marek Chrobak and
                  Lukasz Jez and
                  Dorian Nogneng and
                  Jir{\'{\i}} Sgall},
  editor       = {Chandra Chekuri},
  title        = {Better Approximation Bounds for the Joint Replenishment Problem},
  booktitle    = {Proceedings of the Twenty-Fifth Annual {ACM-SIAM} Symposium on Discrete
                  Algorithms, {SODA} 2014, Portland, Oregon, USA, January 5-7, 2014},
  pages        = {42--54},
  publisher    = {{SIAM}},
  year         = {2014},
  doi          = {10.1137/1.9781611973402.4},
}

@article{BR:12,
  author       = {Carlos Fisch Brito and
                  Elias Koutsoupias and
                  Shailesh Vaya},
  title        = {Competitive Analysis of Organization Networks or Multicast Acknowledgment:
                  How Much to Wait?},
  journal      = {Algorithmica},
  volume       = {64},
  number       = {4},
  pages        = {584--605},
  year         = {2012},
  doi          = {10.1007/S00453-011-9567-5},
}

@inproceedings{Doo:98,
  author       = {Daniel R. Dooly and
                  Sally A. Goldman and
                  Stephen D. Scott},
  editor       = {Jeffrey Scott Vitter},
  title        = {{TCP} Dynamic Acknowledgment Delay: Theory and Practice (Extended
                  Abstract)},
  booktitle    = {Proceedings of the Thirtieth Annual {ACM} Symposium on the Theory
                  of Computing, Dallas, Texas, USA, May 23-26, 1998},
  pages        = {389--398},
  publisher    = {{ACM}},
  year         = {1998},
  doi          = {10.1145/276698.276792},
}

@article{Doo:01,
  author       = {Daniel R. Dooly and
                  Sally A. Goldman and
                  Stephen D. Scott},
  title        = {On-line analysis of the {TCP} acknowledgment delay problem},
  journal      = {J. {ACM}},
  volume       = {48},
  number       = {2},
  pages        = {243--273},
  year         = {2001},
  doi          = {10.1145/375827.375843},
}

@inproceedings{Se:00,
  author       = {Steven S. Seiden},
  editor       = {F. Frances Yao and
                  Eugene M. Luks},
  title        = {A guessing game and randomized online algorithms},
  booktitle    = {Proceedings of the Thirty-Second Annual {ACM} Symposium on Theory
                  of Computing, May 21-23, 2000, Portland, OR, {USA}},
  pages        = {592--601},
  publisher    = {{ACM}},
  year         = {2000},
  url          = {https://doi.org/10.1145/335305.335385},
  doi          = {10.1145/335305.335385},
}

@inproceedings{Ka:01,
  author       = {Anna R. Karlin and
                  Claire Kenyon and
                  Dana Randall},
  editor       = {Jeffrey Scott Vitter and
                  Paul G. Spirakis and
                  Mihalis Yannakakis},
  title        = {Dynamic {TCP} acknowledgement and other stories about e/(e-1)},
  booktitle    = {Proceedings on 33rd Annual {ACM} Symposium on Theory of Computing,
                  July 6-8, 2001, Heraklion, Crete, Greece},
  pages        = {502--509},
  publisher    = {{ACM}},
  year         = {2001},
  doi          = {10.1145/380752.380845},
}

@article{AS:25,
  author       = {Yossi Azar and
                  Shahar Lewkowicz},
  title        = {Online Joint Replenishment Problem with Arbitrary Holding and Backlog
                  Costs},
  journal      = {CoRR},
  volume       = {abs/2507.16096},
  year         = {2025},
  url          = {https://doi.org/10.48550/arXiv.2507.16096},
  doi          = {10.48550/ARXIV.2507.16096},
}

@article{S:25,
  author       = {David B. Shmoys and 
                  Varun Suriyanarayana and
                 Seeun William Umboh},
  title        = {Improved Online Algorithms for Inventory Management Problems with Holding and Delay Costs: Riding the Wave Makes Things Simpler, Stronger, \& More General},
  journal      = {},
  year         = {2025},
}

@article{Al:09,
  author       = {Noga Alon and
                  Baruch Awerbuch and
                  Yossi Azar and
                  Niv Buchbinder and
                  Joseph Naor},
  title        = {The Online Set Cover Problem},
  journal      = {{SIAM} J. Comput.},
  volume       = {39},
  number       = {2},
  pages        = {361--370},
  year         = {2009},
  doi          = {10.1137/060661946},
}

@inproceedings{Bi:13,
  author       = {Marcin Bienkowski and
                  Jaroslaw Byrka and
                  Marek Chrobak and
                  Lukasz Jez and
                  Jir{\'{\i}} Sgall and
                  Grzegorz Stachowiak},
  editor       = {Frank Dehne and
                  Roberto Solis{-}Oba and
                  J{\"{o}}rg{-}R{\"{u}}diger Sack},
  title        = {Online Control Message Aggregation in Chain Networks},
  booktitle    = {Algorithms and Data Structures - 13th International Symposium, {WADS}
                  2013, London, ON, Canada, August 12-14, 2013. Proceedings},
  series       = {Lecture Notes in Computer Science},
  volume       = {8037},
  pages        = {133-145},
  publisher    = {Springer},
  year         = {2013},
  doi          = {10.1007/978-3-642-40104-6\_12},
}

@article{More:Bi:21,
  author       = {Marcin Bienkowski and
                  Martin B{\"{o}}hm and
                  Jaroslaw Byrka and
                  Marek Chrobak and
                  Christoph D{\"{u}}rr and
                  Luk'av{s} Folwarczn'y and
                  Lukasz Jez and
                  Jir{\'{\i}} Sgall and
                  Kim Thang Nguyen and
                  Pavel Vesel{\'{y}}},
  title        = {New results on multi-level aggregation},
  journal      = {Theor. Comput. Sci.},
  volume       = {861},
  pages        = {133-143},
  year         = {2021},
  doi          = {10.1016/J.TCS.2021.02.016},
}

@inproceedings{More:Bo:20,
  author       = {Thomas Bosman and
                  Neil Olver},
  editor       = {Daniel Bienstock and
                  Giacomo Zambelli},
  title        = {Improved Approximation Algorithms for Inventory Problems},
  booktitle    = {Integer Programming and Combinatorial Optimization - 21st International
                  Conference, {IPCO} 2020, London, UK, June 8-10, 2020, Proceedings},
  series       = {Lecture Notes in Computer Science},
  volume       = {12125},
  pages        = {91-103},
  publisher    = {Springer},
  year         = {2020},
  doi          = {10.1007/978-3-030-45771-6\_8},
}

@article{More:Che:16,
  author       = {Maurice Cheung and
                  Adam N. Elmachtoub and
                  Retsef Levi and
                  David B. Shmoys},
  title        = {The submodular joint replenishment problem},
  journal      = {Math. Program.},
  volume       = {158},
  number       = {1-2},
  pages        = {207--233},
  year         = {2016},
  doi          = {10.1007/S10107-015-0920-3},
}

@inproceedings{More:Mar:24,
  author       = {Mathieu Mari and
                  Michal Pawlowski and
                  Runtian Ren and
                  Piotr Sankowski},
  editor       = {Juli{\'{a}}n Mestre and
                  Anthony Wirth},
  title        = {Online Multi-Level Aggregation with Delays and Stochastic Arrivals},
  booktitle    = {35th International Symposium on Algorithms and Computation, {ISAAC}
                  2024, December 8-11, 2024, Sydney, Australia},
  series       = {LIPIcs},
  volume       = {322},
  pages        = {49:1--49:20},
  publisher    = {Schloss Dagstuhl - Leibniz-Zentrum f{\"{u}}r Informatik},
  year         = {2024},
  doi          = {10.4230/LIPICS.ISAAC.2024.49},
}

@article{More:Mc:21,
  author       = {Jeremy McMahan},
  title        = {A D-competitive algorithm for the Multilevel Aggregation Problem with
                  Deadlines},
  journal      = {CoRR},
  volume       = {abs/2108.04422},
  year         = {2021},
  url          = {https://arxiv.org/abs/2108.04422},
  eprinttype    = {arXiv},
  eprint       = {2108.04422},
}

@article{More:Nag:16,
  author       = {Viswanath Nagarajan and
                  Cong Shi},
  title        = {Approximation algorithms for inventory problems with submodular or
                  routing costs},
  journal      = {Math. Program.},
  volume       = {160},
  number       = {1-2},
  pages        = {225--244},
  year         = {2016},
  url          = {https://doi.org/10.1007/s10107-016-0981-y},
  doi          = {10.1007/S10107-016-0981-Y}
}

@inproceedings{AT:20,
  author       = {Yossi Azar and
                  Noam Touitou},
  editor       = {Sandy Irani},
  title        = {Beyond Tree Embeddings - a Deterministic Framework for Network Design
                  with Deadlines or Delay},
  booktitle    = {61st {IEEE} Annual Symposium on Foundations of Computer Science, {FOCS}
                  2020, Durham, NC, USA, November 16-19, 2020},
  pages        = {1368--1379},
  publisher    = {{IEEE}},
  year         = {2020},
  doi          = {10.1109/FOCS46700.2020.00129},
}
